\documentclass{article}
\usepackage{amsmath,amssymb,amsfonts}

\def\op#1{\mathop{{\it\fam0} #1}\limits}

\newcommand{\beq}{\begin{equation}}
\newcommand{\eeq}{\end{equation}}
\newcommand{\ben}{\begin{eqnarray}}
\newcommand{\een}{\end{eqnarray}}
\newcommand{\be}{\begin{eqnarray*}}
\newcommand{\ee}{\end{eqnarray*}}
\newcommand{\bea}{\begin{eqalph}}
\newcommand{\eea}{\end{eqalph}}

\newcommand{\di}{{\mathrm {dim}\,}}

\newcommand{\al}{\alpha}
\newcommand{\bt}{\beta}
\newcommand{\dl}{\delta}
\newcommand{\la}{\lambda}
\newcommand{\La}{\Lambda}
\newcommand{\f}{\phi}

\newcommand{\p}{\pi}

\newcommand{\om}{\omega}
\newcommand{\Om}{\Omega}
\newcommand{\m}{\mu}
\newcommand{\n}{\nu}
\newcommand{\g}{\gamma}
\newcommand{\G}{\Gamma}

\newcommand{\Ker}{\mathrm{Ker}\,}

\newcommand{\thh}{\theta}
\newcommand{\bth}{\mathbf \Theta}
\newcommand{\bom}{\mathbf \Omega}

\newcommand{\vt}{\vartheta}
\newcommand{\cG}{{\mathfrak g}}
\newcommand{\ve}{\varepsilon}

\newcommand{\ap}{\approx}

\newcommand{\nm}[1]{|{#1}|}

\newcommand{\id}{{\mathrm{Id}\,}}

\newcommand{\si}{\sigma}

\newcommand{\cJ}{{\mathcal J}}
\newcommand{\cR}{{\mathcal R}}

\newcommand{\cT}{{\mathcal T}}
\newcommand{\cP}{{\mathcal P}}
\newcommand{\cL}{{\mathcal L}}

\newcommand{\cE}{{\mathcal E}}
\newcommand{\cH}{{\mathcal H}}
\newcommand{\cF}{{\mathcal F}}

\newcommand{\cN}{{\mathcal N}}

\newcommand{\cS}{{\mathcal S}}

\newcommand{\bL}{{\mathbf L}}

\newcommand{\w}{\wedge}

\newcommand{\wt}{\widetilde}
\newcommand{\wh}{\widehat}
\newcommand{\ol}{\overline}

\newcommand{\dr}{\partial}
\newcommand{\x}{\xi}

\newcommand{\llra}{\longleftrightarrow}

\newcommand{\ar}{\op\longrightarrow}

\newcommand{\ot}{\otimes}

\let\ssection=\section
\renewcommand{\section}{\setcounter{equation}{0}\ssection}

\newenvironment{eqalph}{\stepcounter{equation}
\setcounter{equationa}{\value{equation}} \setcounter{equation}{0}
\begin{eqnarray}}{\end{eqnarray}\setcounter{equation}{\value{equationa}}}

\newcounter{equationa}[section]
\newcounter{remark}[section]
\newcounter{example}[section]
\newcounter{theorem}[section]
\newcounter{condition}[section]
\newcounter{lemma}[section]
\newcounter{corollary}[section]
\newcounter{definition}[section]
\def\theremark{\arabic{section}.\arabic{remark}}

\def\thetheorem{\arabic{section}.\arabic{theorem}}

\def\thedefinition{\arabic{section}.\arabic{theorem}}

\newenvironment{proof}{{\it Proof.}}{\hfill {\Large $\bullet$}
\medskip }
\newenvironment{remark}{\refstepcounter{remark} \medskip {\bf Remark
\theremark.} }{\hfill {\Large $\bullet$} \medskip }
\newenvironment{example}{\refstepcounter{remark} \medskip {\bf
Example \theremark.} }{ \hfill {\Large $\bullet$} \medskip }
\newenvironment{theorem}{\refstepcounter{theorem} \medskip{\bf
Theorem \thetheorem.} }{ \hfill $\Box$ \medskip }

\newenvironment{lemma}{\refstepcounter{theorem} \medskip{\bf  Lemma
\thetheorem.}}{\hfill $\Box$ \medskip }
\newenvironment{corollary}{\refstepcounter{theorem} \medskip{\bf
Corollary \thetheorem.} }{ \hfill $\Box$ \medskip }
\newenvironment{definition}{\refstepcounter{theorem} \medskip{\bf
Definition \thedefinition.} }{\hfill $\Box$ \medskip }

\newcommand{\mar}[1]{}

\begin{document}

\hbox{}

\begin{center}

{\Large\bf Polysymplectic Hamiltonian Field Theory}

\bigskip

G. SARDANASHVILY,

\medskip

Department of Theoretical Physics, Moscow State University, Russia

\bigskip

\end{center}

\begin{abstract}
Applied to field theory, the familiar symplectic technique leads
to instantaneous Hamiltonian formalism on an infinite-dimensional
phase space. A true Hamiltonian partner of first order Lagrangian
theory on fibre bundles $Y\to X$ is covariant Hamiltonian
formalism in different variants, where momenta correspond to
derivatives of fields relative to all coordinates on $X$. We
follow polysymplectic (PS) Hamiltonian formalism on a Legendre
bundle over $Y$ provided with a polysymplectic $TX$-valued form.
If $X=\mathbb R$, this is a case of time-dependent
non-relativistic mechanics. PS Hamiltonian formalism is equivalent
to the Lagrangian one if Lagrangians are hyperregular. A
non-regular Lagrangian however leads to constraints and requires a
set of associated Hamiltonians. We state comprehensive relations
between Lagrangian and PS Hamiltonian theories in a case of
semiregular and almost regular Lagrangians. Quadratic Lagrangian
and PS Hamiltonian systems, e.g. Yang -- Mills gauge theory  are
studied in detail. Quantum PS Hamiltonian field theory can be
developed in the frameworks both of familiar functional integral
quantization and quantization of the PS bracket.
\end{abstract}

\bigskip
\bigskip
\bigskip

\tableofcontents

\section*{Introduction}

\addcontentsline{toc}{section}{Introduction}

Applied to field theory, the familiar symplectic Hamiltonian
technique takes the form of instantaneous Hamiltonian formalism on
an infinite-dimensional phase space, where canonical coordinates
are field functions at some instant of time \cite{got91a}. The
true Hamiltonian counterpart of classical first order Lagrangian
field theory on a fibre bundle $Y\to X$ is covariant Hamiltonian
formalism, where canonical momenta $p^\m_i$ correspond to jets
$y^i_\m$ of field variables $y^i$ with respect to all coordinates
$(x^\m)$ on a base $X$. This formalism has been vigorously
developed since 1970s in the Hamilton -- De Donder,
polysymplectic, multisymplectic, $k$-symplectic, $k$-cosymplectic
and other variants (see
\cite{cantr,cari,ech,forg05,jpa99,got,gun,helein,2kij,krupk2,leon2,leon,marsd,rey,rr,rossi,sardz93,zakh}
and references therein).

Here, we are not concerned with higher order Lagrangian  theory
\cite{camp,priet,saun,vitagl} and poly-Poisson formalism
\cite{poly1,poly2}.

We follow polysymplectic (PS) Hamiltonian formalism where the
Legendre bundle $\Pi$ (\ref{00}) plays a role of the momentum
phase space of field theory on a fibre bundle $Y\to X$
\cite{book,jpa99,book09,sardz93,sard94,book95}. It is provided
with the canonical $TX$-valued PS form (\ref{2.4}) and, if $Y\to
X$ is a vector bundle, with the PS bracket (\ref{xx3}) (Section
4).

If $X=\mathbb R$, this is the case of Hamiltonian time-dependent
(non-autonomous) non-relativistic mechanics on a fibre bundle
$Q\to\mathbb R$ \cite{book10,sard98,sard13}. Its momentum phase
space is the vertical cotangent bundle $V^*Q$ of $Q\to\mathbb R$
endowed with the vertical (fibrewise) Poisson structure
(\ref{m72}) (Section 7).

A Hamiltonian in PS theory on the Legendre bundle $\Pi$ (\ref{00})
is defined as a section $h$ of the one-dimensional affine bundle
$Z_L\to \Pi$ (\ref{b418'}) where $Z_Y$ (\ref{N41}) is a
homogeneous Legendre bundle endowed with the canonical
multisymplectic form (\ref{ps13})). Its pull-back with respect to
a Hamiltonian $-h$ is a Hamiltonian form $H$ (\ref{b418}) on a
Legendre bundle $\Pi$. It defines the first order covariant
Hamilton equation (\ref{b4100a}) -- (\ref{b4100b}) on $\Pi$.

The Legendre bundle $\Pi$ (\ref{00}) and the homogeneous Legendre
bundle $Z_Y$ (\ref{N41}) play a role of the momentum and
homogeneous momentum phase spaces in PS Hamiltonian theory,
respectively (Sections 5 -- 6).

Any first order Lagrangian $L$ (\ref{23f2}) on a jet  manifold
$J^1Y$ yields the Legendre map $\wh L: J^1Y\to \Pi$ (\ref{b330}).
Conversely, any Hamiltonian form $H$ (\ref{b418}) on $\Pi$ defines
the momentum map $\wh H: \Pi\to J^1Y$ (\ref{415}). With the
Legendre and momentum maps, on can relate Lagrangian formalism on
$J^1Y$ and PS Hamiltonian formalism on $\Pi$. They are not
equivalent in general.

Lagrangian and PS Hamiltonian formalisms are equivalent in a case
of a hyperregular Lagrangian $L$ (Theorem \ref{ps25}). In Section
9, we state the comprehensive relations between Lagrangian and PS
theories in a case of semiregular and almost regular Lagrangians
(Definition \ref{d11}).

It should be emphasized that PS Hamiltonian formalism on a
Legendre bundle $\Pi$ is equivalent to particular first order
Lagrangian theory on $\Pi$. Given the Hamiltonian form $H$
(\ref{b418}) on $\Pi$, the corresponding covariant Hamilton
equation (\ref{b4100a}) -- (\ref{b4100b}) is the Euler -- Lagrange
equation (\ref{b327}) of the affine first order Lagrangian $L_H$
(\ref{Q3}) (Theorem \ref{91t3}). Moreover, it follows from the
equality (\ref{b4180}) that a Hamiltonian form $H$ possesses the
same classical symmetries as a Lagrangian $L_H$ (Section 10). This
fact enable us to describe symmetries of PS Hamiltonian theory
similarly to those in Lagrangian formalism.

In Section 11, the vertical extension of Lagrangian and PS
Hamiltonian systems are considered. They describe linear
deviations of solutions of Euler -- Lagrange and covariant
Hamilton equations which are Jacobi fields.

Section 12 provides the detailed analysis of quadratic Lagrangian
and PS Hamiltonian systems. The most physically relevant example
of these systems is Yang -- Mills gauge theory of principal
connections (Section 13). Its analysis shows that the main
ingredients in gauge theory are not directly related  with the
gauge invariance property, but are common for theories with almost
regular quadratic Lagrangians.

Affine Lagrangian and PS Hamiltonian systems also are considered
(Section 14). For instance, this is the case of metric-affine
gravitation theory \cite{book,sard94,book95}.

In order to quantize covariant Hamiltonian field theory, one
usually attempts to construct different multimomentum
generalizations of a Poisson bracket
\cite{forg05,forg,helein,kanatch97,kanatch15}. In Section 16, we
discuss some variants of such kind quantization based on the PS
bracket (\ref{xx3}).

At the same time, the above mentioned fact that PS Hamiltonian
formalism on a Legendre bundle $\Pi$ is equivalent to particular
first order Lagrangian theory on $\Pi$ enables us to quantize PS
Hamiltonian field theory in the framework of familiar perturbative
quantum field theory (Section 15).

Throughout the work, we follow familiar technique of fibre
bundles, jet manifolds and connections \cite{book00,book13,sau}.
All morphisms are smooth, and manifolds are smooth real and
finite-dimensional. Smooth manifolds customarily are assumed to be
Hausdorff second-countable and, consequently, locally compact and
paracompact. Unless otherwise stated, they are connected.

The standard symbols $\otimes$, $\vee$, and $\wedge$ stand for the
tensor, symmetric, and exterior products, respectively. The
interior product (contraction) is  denoted by $\rfloor$. By
$\dr^A_B$  are meant the partial derivatives with respect to
coordinates with indices $^B_A$.

If $Z$ is a manifold, we denote by
\be
\pi_Z:TZ\to Z, \qquad \pi^*_Z:T^*Z\to Z
\ee
its tangent and cotangent bundles, respectively. Given manifold
coordinates $(z^\al)$ on $Z$, they are equipped with the holonomic
coordinates
\be
(z^\la,\dot z^\la), \quad \dot z'^\la= \frac{\dr z'^\la}{\dr
z^\mu}\dot z^\m, \qquad (z^\la,\dot z_\la), \quad \dot z'_\la=
\frac{\dr z'^\m}{\dr z^\la}\dot z_\m,
\ee
with respect to the holonomic frames $\{\dr_\la\}$ and  coframes
$\{dz^\la\}$ in the tangent and cotangent spaces to $Z$,
respectively. Any manifold morphism $f:Z\to Z'$ yields the tangent
morphism
\be
Tf:TZ\to TZ', \qquad \dot z'^\la\circ Tf = \frac{\dr f^\la}{\dr
x^\m}\dot z^\m.
\ee
We use the compact notation $\dot\dr_\m = \dr/\dr\dot z^\mu.$

Given a fibre bundle $Y\to X$ endowed with bundle coordinates
$(x^\la,y^i)$, we denote by $VY$ and $V^*$ its vertical tangent
and cotangent bundles provided with holonomic coordinates
$(x^\la,y^i,\dot y^i)$ and $(x^\la, y^i, \ol y_i)$, respectively.

The symbol $C^\infty(Z)$ stands for the ring of smooth real
functions on a manifold $Z$.

\section{First order Lagrangian formalism on fibre bundles}

As was mentioned above, we restrict our consideration to first
order Lagrangian formalism on a smooth fibre bundle $\pi:Y\to X$
over an oriented $(1<n)$-dimensional base $X$ (see Section 7 for a
case of $n=1$) \cite{book,book09,book13}. Let $Y$ be provided with
an atlas of bundle coordinates $(x^\la, y^i)$.

A velocity phase space of first order Lagrangian theory on a fibre
bundle $Y$ is the first order jet manifold $J^1Y$ of sections of
$Y\to X$ (or, simply, of $Y$). It is endowed with the adapted
coordinates $(x^\la, y^i, y^i_\la)$ possessing transition
functions
\mar{50}\beq
{y'}^i_\la = \frac{\dr x^\m}{\dr{x'}^\la}(\dr_\m
+y^j_\m\dr_j)y'^i.\label{50}
\eeq
There are natural fibrations
\mar{1.15}\beq
\pi^1:J^1Y\to X, \qquad \pi^1_0:J^1Y\to Y, \label{1.15}
\eeq
where the latter is an affine bundle modelled over a vector bundle
\mar{cc9}\beq
T^*X \op\otimes_Y VY\to Y.\label{cc9}
\eeq
There are the canonical imbeddings
\mar{18,24}\ben
&&\la_{(1)}:J^1Y\op\to_Y
T^*X \op\otimes_Y TY,\qquad \la_{(1)}=dx^\la
\otimes (\dr_\la + y^i_\la \dr_i)=dx^\la\otimes d_\la, \label{18}\\
&&\thh_{(1)}:J^1Y \op\to_Y T^*Y\op\otimes_Y VY,\qquad
\thh_{(1)}=(dy^i- y^i_\la dx^\la)\otimes \dr_i=\thh^i \otimes
\dr_i,\label{24}
\een
where $d_\la$ denote the total derivatives, and $\thh^i$ are local
contact forms.

A first order Lagrangian of Lagrangian theory on a fibre bundle
$Y\to X$ is defined as a density
\mar{23f2}\beq
L=\cL\om: J^1Y\to \op\w^n T^*X \label{23f2}
\eeq
on the first order jet manifold $J^1Y$ of $Y$ where, for the sake
of simplicity, we denote
\be
\op\w^n T^*X=\pi_1^*(\op\w^n T^*X)=J^1Y\op\times_X \op\w^n
T^*X=J^1Y\op\times_Y \op\w^n T^*X.
\ee

The corresponding second-order Euler -- Lagrange operator reads
\mar{305}\ben
&& \dl L=\cE_L: J^2Y\to T^*Y\w(\op\w^nT^*X), \nonumber \\
&& \cE_L= (\dr_i\cL- d_\la\pi^\la_i) \thh^i\w\om, \qquad
\pi^\la_i=\dr^\la_i\cL,  \label{305}\\
&& d_\la=\dr_\la +y^i_\la\dr_i +y^i_{\la\m}\dr^\m_i, \nonumber
\een
where $J^2Y$ is the second order jet manifold of $Y\to X$
coordinated by $(x^\la,y^i,y^i_\la, y^i_{\la\m}=y^i_{\m\la})$. Its
kernel $\Ker\cE_L\subset J^2Y$ is the second order Euler --
Lagrange equation on $Y$ locally given by equalities
\mar{b327}\beq
\dl_i L(\dr_i- d_\la\dr^\la_i)\cL=0. \label{b327}
\eeq

\begin{remark} Strictly speaking, the Euler -- Lagrange equation
(\ref{b327}) fails to be a differential equation in general
because $\Ker\cE_L\subset J^2Y$ need not be a closed subbundle of
$J^2Y\to X$ \cite{book,book09}.
\end{remark}

In a general setting, Lagrangians (\ref{23f2}) and Euler --
Lagrange operators $\dl L$ (\ref{305}) in Lagrangian formalism on
a fibre bundle $Y\to X$ are introduced as elements of the
variational bicomplex of exterior forms on an infinite order jet
manifold $J^\infty Y$ \cite{cmp05,book09,book13}. Its cohomology
defines the global variational decomposition
\mar{+421}\beq
dL=\dl L-d_H \Xi_L \label{+421}
\eeq
over $J^2Y$ where
\mar{ps6}\beq
d_H\phi=dx^\la\w d_\la\phi \label{ps6}
\eeq
is the total differential of exterior forms $\phi$ on $J^2Y$, and
$\Xi_L$ is some Lepage form which is a Lepage equivalent of a
Lagrangian $L$, i.e.,
\be
&& L=h_0(\Xi_L), \\
&& h_0(dx^\la)= dx^\la,\quad h_0(dy^i)=y^i_\la dx^\la, \quad
h_0(dy^i_\m)=y^i_{\la\m} dx^\la.
\ee
Defined up to a $d_H$-closed form, a form $\Xi_L$ reads
\mar{22f44}\beq
\Xi_L=L+(\p^\la_i-d_\m \si^{\m\la}_i)\thh^i\w\om_\la
+\si^{\la\m}_i \thh^i_\m\w\om_\la, \qquad
\om_\la=\dr_\la\rfloor\om, \label{22f44}
\eeq
where $\si^{\m\la}_i=-\si^{\la\m}_i$ are skew-symmetric local
functions on $Y$. Lepage equivalents constitute an affine space
modelled over a vector space of $d_H$-exact one-contact Lepage
forms
\be
\rho= -d_\m \si^{\m\la}_i\thh^i\w\om_\la +\si^{\la\m}_i
\thh^i_\m\w\om_\la.
\ee
Let us choose the Poincar\'e -- Cartan form
\mar{303}\beq
H_L=\cL\om +\p^\la_i\thh^i\w\om_\la  \label{303}
\eeq
as the origin of this affine space because it is defined on
$J^1Y$.

Given the Lagrangian $L$ (\ref{23f2}), let us consider the
vertical tangent map
\mar{ps2}\beq
 VL: V_YJ^1Y\op\ar_Y V_YJ^1Y\op\times_Y \op\w^n T^*X \label{ps2}
\eeq
to $L$ over $Y$, where $V_YJ^1Y$ denotes the vertical tangent
bundle of $J^1Y\to Y$. Since $J^1Y\to Y$ is an affine
 bundle modelled over the vector bundle (\ref{cc9}), we have the canonical vertical splitting
\be
V_YJ^1Y=J^1Y\op\times_Y (T^*X \op\otimes_Y VY).
\ee
Accordingly, the vertical tangent map $VL$ (\ref{ps2}) yields a
linear morphism
\be
J^1Y\op\times_Y(T^*X\op\ot_YVY) \ar J^1Y\op\times_Y (\op\w^n T^*X)
\ee
over $J^1Y$ and the corresponding morphism
\mar{b330}\beq
\wh L:J^1Y \to V^*Y\op\ot_Y(\op\w^nT^*X)\op\ot_YTX \label{b330}
\eeq
over $Y$. It is called the Legendre map associated to a Lagrangian
$L$.

A fibre bundle
\mar{00}\beq
\Pi =V^*Y\op\ot_Y(\op\w^nT^*X)\op\ot_YTX\op\ar^{\pi_{\Pi Y}} Y
\label{00}
\eeq
over $Y$ is called the Legendre bundle. It is provided with the
holonomic coordinates $(x^\la, y^i, p^\la_i)$ possessing
transition functions
\mar{2.3}\beq
{p'}^\la_i = \det \left(\frac{\dr x^\ve}{\dr {x'}^\nu}\right)
\frac{\dr y^j}{\dr{y'}^i} \frac{\dr {x'}^\la}{\dr x^\m}p^\m_j.
\label{2.3}
\eeq
With respect to these coordinates, the Legendre map (\ref{b330})
reads
\mar{m11}\beq
p^\la_i\circ\wh L =\pi^\la_i. \label{m11}
\eeq

\begin{remark}
There is the canonical isomorphism
\mar{000}\beq
\Pi= V^*Y\op\w_Y(\op\w^{n-1} T^*X), \qquad (p^\la_i)\to p^\la_i
\ol dy^i\om_\la, \label{000}
\eeq
where $\{\ol dy^i\}$ are fibre bases for the vertical cotangent
bundle $V^*Y$ of $Y\to X$.
\end{remark}

Certainly, the Legendre map (\ref{b330}) need not be a bundle
isomorphism. Its range
\mar{ps3}\beq
N_L=\wh L(J^1Y)\subset \Pi \label{ps3}
\eeq
is called the Lagrangian constraint space.

\begin{definition} \label{d11} \mar{d11}
A Lagrangian $L$ is said to be:

$\bullet$ hyperregular if the Legendre map $\wh L$ is a
diffeomorphism;

$\bullet$ regular if $\wh L$ is a local diffeomorphism over $Y$,
i.e., $\det(\dr^\m_i\dr^\n_j\cL)\neq 0$;

$\bullet$ semiregular if the inverse image $\wh L^{-1}(z)$ of any
point $z\in N_L$ is a connected submanifold of $J^1Y$;

$\bullet$ almost regular if the Lagrangian constraint space $N_L$
(\ref{ps3}) is a closed imbedded subbundle
\mar{ps4}\beq
i_N:N_L\to\Pi \label{ps4}
\eeq
of a Legendre bundle $\Pi\to Y$ and the Legendre map
\mar{cmp12}\beq
\wh L:J^1Y\to N_L \label{cmp12}
\eeq
is a fibred manifold with connected fibres (i.e., a Lagrangian $L$
is semiregular).
\end{definition}

The Poincar\'e -- Cartan form (\ref{303}) in turn takes its values
into a subbundle
\be
J^1Y\op\times_Y (T^*Y\op\w_Y (\op\w^{n-1}T^*X))
\ee
of $\op\w^n T^*J^1Y$. Hence, it defines a bundle morphism
\mar{N41}\beq
\wh H_L: J^1Y\to Z_Y=T^*Y\op\w_Y (\op\w^{n-1}T^*X), \label{N41}
\eeq
over $Y$ whose range
\mar{23f10}\beq
Z_L= \wh H_L(J^1Y) \label{23f10}
\eeq
is an imbedded subbundle $i_L:Z_L\to Z_Y$ of the fibre bundle
$Z_Y\to Y$ (\ref{N41}). This morphism is called the homogeneous
Legendre map.
 Accordingly, the fibre bundle $Z_Y\to Y$ (\ref{N41}) is said to
be the homogeneous Legendre bundle.  It is equipped with holonomic
coordinates $(x^\la,y^i,p^\la_i,p)$ possessing transition
functions
\mar{2.3'}\beq
{p'}^\la_i = \det \left(\frac{\dr x^\ve}{\dr {x'}^\nu}\right)
\frac{\dr y^j}{\dr{y'}^i} \frac{\dr {x'}^\la}{\dr x^\m}p^\m_j,
\qquad p'=\det \left(\frac{\dr x^\ve}{\dr {x'}^\nu}\right)
\left(p-\frac{\dr y^j}{\dr {y'}^i}\frac{\dr {y'}^i}{\dr
x^\mu}p^\mu_j\right). \label{2.3'}
\eeq
With respect to these coordinates, the morphism $\wh H_L$
(\ref{N41}) reads
\be
(p^\m_i, p)\circ \wh H_L =(\p^\m_i, \cL-y^i_\m\p^\m_i).
\ee

A glance at the transition functions (\ref{2.3'}) shows that $Z_Y$
(\ref{N41}) is a one-dimensional affine bundle
\mar{b418'}\beq
\pi_{Z\Pi}:Z_Y\to \Pi\label{b418'}
\eeq
over the Legendre bundle $\Pi$ (\ref{000}) modelled over the
pull-back vector bundle
\mar{ps12}\beq
\Pi\op\times_X\op\w^nT^*X\to \Pi. \label{ps12}
\eeq
Moreover, the Legendre map $\wh L$ (\ref{b330}) is exactly the
composition of morphisms
\mar{m11'}\beq
\wh L=\pi_{Z\Pi}\circ \wh H_L:J^1Y \op\to_Y \Pi. \label{m11'}
\eeq

As was mentioned above, the Legendre bundle $\Pi$ (\ref{00}) and
the homogeneous Legendre bundle $Z_Y$ (\ref{N41}) play a role of
the momentum phase space and a homogeneous momentum phase space in
PS Hamiltonian theory, respectively (Section 5).

\section{Cartan and Hamilton -- De Donder equations}

Note that the Euler -- Lagrange equation (\ref{b327}) do not
exhaust all equations considered in first order Lagrangian theory.

Being a Lepage equivalent of a Lagrangian $L$, the Poincar\'e --
Cartan form $H_L$ (\ref{303}) also is a Lepage equivalent of a
first order Lagrangian
\mar{cmp80}\beq
\ol L=\wh h_0(H_L) = (\cL + (\wh y_\la^i - y_\la^i)\p_i^\la)\om,
\qquad \wh h_0(dy^i)=\wh y^i_\la dx^\la, \label{cmp80}
\eeq
on the repeated jet manifold $J^1J^1Y$, coordinated by
$(x^\la,y^i,y^i_\la,\wh y_\m^i,y^i_{\m\la})$.

The Euler -- Lagrange operator for $\ol L$ (called the Euler --
Lagrange -- Cartan operator) reads
\mar{2237}\ben
&& \dl\ol L=\cE_{\ol L} : J^1J^1Y\to T^*J^1Y\w(\op\w^n T^*X), \nonumber \\
&& \cE_{\ol L} = [(\dr_i\cL - \wh d_\la\p_i^\la
+ \dr_i\p_j^\la(\wh y_\la^j - y_\la^j))dy^i + \dr_i^\la\p_j^\m(\wh
y_\m^j - y_\m^j) dy_\la^i]\w
\om, \label{2237} \\
&&\wh d_\la=\dr_\la +\wh y^i_\la\dr_i +
y^i_{\la\m}\dr_i^\m. \nonumber
\een
Its kernel $\Ker\cE_{\ol L}\subset J^1J^1Y$ is the first order
Cartan equation on $J^1Y$ locally given by equalities
\mar{b336}\ben
&& \dr_i^\la\p_j^\m(\wh y_\m^j - y_\m^j)=0, \label{b336a}\\
&& \dr_i \cL - \wh d_\la\p_i^\la
+ (\wh y_\la^j - y_\la^j)\dr_i\p_j^\la=0. \label{b336b}
\een

Since $\cE_{\ol L}|_{J^2Y}=\cE_L$, the Cartan equation
(\ref{b336a})  --  (\ref{b336b}) is equivalent to the Euler --
Lagrange on (\ref{b327}) on integrable sections $\ol s=J^1s$ of
$J^1Y\to X$, where $s$ are sections of $Y\to X$. These equations
are equivalent if a Lagrangian is regular. The Cartan equation
(\ref{b336a})  --  (\ref{b336b}) on sections $\ol s: X\to J^1Y$ is
equivalent to the relation
\mar{C28}\beq
\ol s^*(u\rfloor dH_L)=0, \label{C28}
\eeq
which is assumed to hold for all vertical vector fields $u$ on
$J^1Y\to X$.

The homogeneous Legendre bundle $Z_Y$ (\ref{N41}) admits the
canonical multisymplectic Liouville form
\mar{N43}\beq
\Xi_Y= p\om + p^\la_i dy^i\w\om_\la. \label{N43}
\eeq
Accordingly, its  imbedded subbundle $Z_L$ (\ref{23f10}) is
provided with the pull-back De Donder form $\Xi_L=i^*_L\Xi_Y$.
There is the equality
\mar{cmp14}\beq
H_L=\wh H_L^*\Xi_L=\wh H_L^*(i_L^*\Xi_Y).  \label{cmp14}
\eeq
By analogy with the Cartan equation (\ref{C28}), the Hamilton --
De Donder equation for sections $\ol r$ of $Z_L\to X$ is written
as
\mar{N46}\beq
\ol r^*(u\rfloor d\Xi_L)=0, \label{N46}
\eeq
where $u$ is an arbitrary vertical vector field on $Z_L\to X$.
Then the following holds \cite{got}.

\begin{theorem}\label{ddd} \mar{ddd} Let the homogeneous Legendre map
$\wh H_L$ be a submersion. Then a section $\ol s$ of $J^1Y\to X$
is a solution of the Cartan equation (\ref{C28}) iff $\wh
H_L\circ\ol s$ is a solution of the Hamilton -- De Donder equation
(\ref{N46}), i.e., the Cartan and Hamilton -- De Donder equations
are quasi-equivalent.
\end{theorem}

Note that the Cartan and Hamilton -- De Donder equations play a
role of the Lagrangian partners of Hamilton equations in PS
Hamiltonian theory (Section 9).

\section{Polysymplectic structure}

Treated as a momentum phase space of fields, the Legendre bundle
$\Pi$ (\ref{00}) is endowed with the following polysymplectic (PS)
structure.

There is the canonical bundle monomorphism
\mar{2.4}\beq
\bth_Y :\Pi\op\ar_Y\op\w^{n+1}T^*Y\op\otimes_Y TX, \qquad  \bth_Y
=p^\la_idy^i\w\om\otimes\dr_\la, \label{2.4}
\eeq
called the tangent-valued Liouville form on $\Pi$. Strictly
speaking, it is $TX$-valued, but not a tangent-valued (i.e.,
$T\Pi$-valued) form on $\Pi$. Therefore, standard technique of
tangent-valued forms, as like as that of exterior forms, is not
applied to $\bth_Y$ (\ref{2.4}).

 At the same
time, there is a unique $TX$-valued $(n+2)$-form
\mar{406}\beq
\bom_Y =dp_i^\la\w dy^i\w \om\ot\dr_\la \label{406}
\eeq
on $\Pi$ such that the relation
\mar{ps30}\beq
\bom_Y\rfloor\f =d(\bth_Y\rfloor\f) \label{ps30}
\eeq
holds for an arbitrary exterior one-form $\f$ on $X$. The form
(\ref{406}) is called the polysymplectic (PS) form.

\begin{remark}
It should be emphasized that, following \cite{gun}, one often
provides $\Pi$ (\ref{00}) with an exterior form
\be
dp_i^\la\w dy^i\w \om_\la, \qquad \om_\la=\dr_\la\rfloor\om,
\ee
which however globally is ill defined because it is not maintained
under the transition functions (\ref{2.3}) (cf. the form
(\ref{000})). Our variant of PS formalism on the Legendre bundle
$\Pi$ (\ref{00}) is based just on the $TX$-valued PS form $\bom_Y$
(\ref{406}) \cite{book,jpa99,book09,sardz93,book95}.
\end{remark}

Let $J^1\Pi$ be the first order jet manifold of a fibre bundle
$\Pi\to X$. It is equipped with the adapted coordinates $(x^\la
,y^i,p^\la_i,y^i_\m,p^\la_{\m i})$. A connection
\mar{cmp33}\beq
\g =dx^\la\otimes(\dr_\la +\g^i_\la\dr_i +\g^\m_{\la i}\dr^i_\m)
\label{cmp33}
\eeq
on $\Pi\to X$ is called the Hamiltonian connection if an exterior
form
\be
\g\rfloor\bom_Y=(\dr_\la +\g^i_\la\dr_i +\g^\m_{\la
i}\dr^i_\m)\rfloor (dp_i^\la\w dy^i\w \om)
\ee
on $J^1\Pi$ is closed. Components of a Hamiltonian connection
satisfy the conditions
\mar{b422}\beq
\dr^i_\la\g^j_\m-\dr^j_\m\g^i_\la=0,\qquad
  \dr_i\g_{\m j}^\m- \dr_j\g_{\m i}^\m=0,
\qquad \dr_j\g_\la^i+\dr_\la^i\g_{\m j}^\m=0. \label{b422}\\
\eeq

If a form $\g\rfloor\bom_Y$ is closed, there is a contractible
neighborhood $U$ of each point of $\Pi$ where a local form
$\g\rfloor\bom_Y$ is exact, i.e.,
\mar{cmp4}\beq
\g\rfloor\bom_Y =dp^\la_i\w dy^i\w\om_\la - (\g^i_\la dp^\la_i
-\g^\la_{\la i}dy^i)\w\om =dH_U \label{cmp4}
\eeq
on $U$. It is readily observed that, by virtue of the conditions
(\ref{b422}), the second term in the right-hand side of this
equality also is closed and, consequently, an exact form on $U$.
In accordance with the relative Poincar\'e lemma (\cite{book},
Remark 4.4.2), this term can be brought into the form
$d\cH_U\w\om$ where $\cH_U$ is a local function on $U$. Then a
form $H_U$ in the expression (\ref{cmp4}) reads
\mar{mos06}\beq
H_U=p^\la_idy^i\w\om_\la -\cH_U\om. \label{mos06}
\eeq

\begin{example}\label{e43.5} Every connection
\mar{ps9}\beq
\G=dx^\la\ot(\dr_\la +\G^i_\la\dr_i) \label{ps9}
\eeq
on a fibre bundle $Y\to X$ gives rise to a connection
\mar{404}\ben
\wt\G =dx^\la\otimes[\dr_\la +\G^i_\la (y)\dr_i + (-\dr_j\G^i_\la
p^\m_i+ K_\la{}^\m{}_\nu p^\nu_j- K_\la{}^\al{}_\al
p^\m_j)\dr^j_\m]  \label{404}
\een
on a Legendre bundle $\Pi\to X$, where $K$ is a symmetric linear
connection
\mar{08}\beq
K= dx^\la\otimes (\dr_\la +K_\la{}^\m{}_\n \dot x^\n
\frac{\dr}{\dr\dot x^\m}), \label{08}
\eeq
on the tangent bundle $TX\to X$. Due to the isomorphism
(\ref{000}), the connection (\ref{404}) is constructed as follows
\cite{book}. It is a tensor product
\beq
\wt\G=(\G\times K)\op\ot_\G V^*\G \label{b422'}
\eeq
over $\G$  of the product connection $\G\times K$ on the pull-back
bundle
\be
Y\op\times_X\op\w^{n-1}T^*X\to X
\ee
and the covertical connection $V^*\G$ to $\G$:
\mar{44}\beq
V^*\G
=dx^\la\otimes(\dr_\la +\G^i_\la\frac{\dr}{\dr y^i}-\dr_j\G^i_\la
\ol y_i \frac{\dr}{\dr \ol y_j}), \label{44}
\eeq
on the vertical cotangent bundle $V^*Y\to X$. Since the
connections $\G\times K$ and $V^*\G$ are linear connections over
$\G$, their tensor product (\ref{b422'}) is well defined. The
connection (\ref{404}) on $\Pi\to X$, by construction, projects
onto the connection $\G$ on $Y\to X$. It obeys a relation
\be
\wt\G\rfloor\bom_Y =-d(\G\rfloor\bth_Y)=dH_\G
\ee
where a form
\mar{3.6}\beq
H_\G =\G^*\Xi_Y =p^\la_i dy^i\w\om_\la -p^\la_i\G^i_\la\om
\label{3.6}
\eeq
globally is well defined. It follows that $\wt\G$ (\ref{404}) is a
Hamiltonian connection.
\end{example}

Thus, Hamiltonian connections always exist on a Legendre bundle
$\Pi\to X$, and every connection $\G$ on $Y\to X$ gives rise to a
Hamiltonian connection on $\Pi\to X$.

\section{PS bracket}

If $Y\to X$ is a vector bundle, the Legendre bundle $\Pi$
(\ref{00}) also is provided with a PS bracket. Note that different
generalizations of a  Poisson bracket have been suggested in the
framework of covariant Hamiltonian field theory
\cite{forg05,forg,helein,kanatch97,kanatch15,mang99}. Here, we
consider such a bracket in the framework of PS geometry, but it
differs from that in our work \cite{mang99}. In a case of
time-dependent mechanics it reduces to the familiar vertical
Poisson bracket $\{,\}_V$ (\ref{m72}) \cite{book10,sard13}.

Given an exterior bundle
\be
\w TX= (X\times \mathbb R)\op\oplus_X TX\op\oplus_X\op\w^2
TX\op\oplus_X\cdots \op\oplus_X\op\w^nTX
\ee
let us consider a fibre bundle
\mar{ps10}\beq
\Pi\op\times_X (\w TX\ot\op\w^n T^*X)\to \Pi. \label{ps10}
\eeq
One can think of its sections
\mar{ps11}\beq
F=\frac{1}{k!}F^{\m_1\ldots\m_k}\dr_{\m_1}\w\cdots\w\dr_{\m_k}\ot\om,
\qquad |F|=k, \label{ps11}
\eeq
as being $TX$-multivalued densities on $\Pi$. Let $\cT^*(\Pi)$
denote their real space. It is an exterior algebra with respect to
the exterior product
\be
&& F\w G=
\frac{1}{k!q!}F^{\m_1\ldots\m_k}G^{\nu_1\ldots\al_q}\dr_{\m_1}\w\cdots\w\dr_{\m_k}\w
\dr_{\al_1}\w\cdots\w\dr_{\al_q}\ot\om, \\
&& |F\w G|=|F|+|G|, \qquad F\w G=(-1)^{|F||G|}G\w F.
\ee

A manifested PS bracket on elements $F$ (\ref{ps11}) of an algebra
$\cT^*\Pi)$ is introduced by the law
\mar{xx3}\ben
&& \{F,G\}_{PS}= \frac{1}{(k-1)!q!} \frac{\dr F^{\m\m_2\ldots\m_k}}{\dr
p^\m_i} \frac{\dr G^{\al_1\ldots\al_q}}{\dr y^i}
\dr_{\m_2}\w\cdots\w\dr_{\m_k}\w
\dr_{\al_1}\w\cdots\w\dr_{\al_q}\ot\om -\nonumber\\
&& \qquad \frac{1}{(q-1)!k!} \frac{\dr G^{\al\al_2\ldots\al_q}}{\dr p^\al_j}\frac{\dr F^{\m_1\ldots\m_k}}{\dr
y^j}
\dr_{\al_2}\w\cdots\w\dr_{\al_q}\w\dr_{\m_1}\w\cdots\w\dr_{\m_k}\ot\om.
\label{xx3}
\een
If $Y\to X$ is a vector bundle, it is maintained under the
transformations (\ref{2.3}) and, consequently, globally is well
defined.

The PS bracket (\ref{xx3}) obeys the relations
\be
\{F,G\}_{PS}= -\{G,F\}_{PS}, \qquad |\{F,G\}_{PS}|=|F|+|G|-1.
\ee
The Jacobi identity
\be
\{S,\{G,F\}\}_{PS} +\{G,\{F,S\}\}_{PS} +\{F,\{S,D\}\}_{PS}=0
\ee
also holds if all the products
\be
|S-1||G-1|, \qquad |G-1||F-1|, |F-1||S-1|
\ee
are even. In particular, it follows that a space of $TX$-valued
densities $F\in\cT^1(\Pi)$ constitute a real Lie algebra with
respect to the PS bracket $\{,\}_{PS}$ (\ref{xx3}).

If $\di X=1$, the PS bracket (\ref{xx3}) always is defined and, as
was mentioned above, it is reduced to the Poisson bracket
$\{,\}_V$ (\ref{m72}) of the canonical vertical Poisson structure
on the vertical cotangent bundle $V^*Y$ (Section 7).

It should be emphasized that the PS bracket $\{F,G\}_{PS}$
(\ref{xx3}), as like as the Poisson bracket $\{,\}_V$ (\ref{m72}),
fails to describe dynamics of Hamiltonian systems \cite{mang99},
but in particular it yields the bracket (\ref{ps45}) of Noether
Hamiltonian currents. Similarly, the Poisson bracket $\{,\}_V$
(\ref{m72}) defines a Lie bracket of Noether currents in mechanics
\cite{book10}.

The PS bracket (\ref{xx3}) can be applied to quantization of
Hamiltonian field systems (Section 16).

\section{Hamiltonian forms}

In the framework of PS formalism, dynamics of sections of the
Legendre bundle $\Pi$ (\ref{00}) is described in terms of
Hamiltonian forms.

Let us consider the homogeneous Legendre bundle $Z_Y$ (\ref{N41})
and the affine bundle $Z_Y\to\Pi$ (\ref{b418'}) modelled over the
pull-back bundle (\ref{ps12}). The homogeneous Legendre bundle
$Z_Y$ is provided with the canonical multisymplectic Liouville
form $\Xi_Y$ (\ref{N43}). Its exterior differential $d\Xi_Y$ is
the canonical multisymplectic form
\mar{ps13}\beq
\Om_Y=dp\w\om +dp^\la_i\w dy^i\w \om_\la. \label{ps13}
\eeq

Let $h=-\cH\om$ be a section the affine bundle $Z_Y\to\Pi$
(\ref{b418'}). A glance at the transformation law (\ref{2.3'})
shows that it is not a density. By analogy with Hamiltonian
time-dependent mechanics (Section 7), $-h$ is said to be the
covariant Hamiltonian of PS Hamiltonian theory. It defines the
pull-back
\mar{b418}\beq
 H=h^*\Xi_Y= p^\la_i dy^i\w \om_\la -\cH\om \label{b418}
\eeq
of the multisimplectic Liouville form $\Xi_Y$ onto a Legendre
bundle $\Pi$ which is called the Hamiltonian form on $\Pi$.

The following is a straightforward corollary of this definition.

\begin{theorem} \label{91t1} \mar{91t1}
(i) Every connection $\G$ (\ref{ps9}) on a fibre bundle $Y\to X$
yields a section
\be
\G=\ol dy^i\to dy^i- \G^i_\la dx^\la
\ee
of a fibre bundle $T^*Y\to V^*Y$ which gives rise to a section
\be
 h_\G : p^\la_i\ol dy^i\ot\om_\la \to  p^\la_i dy^i\w\om_\la -
p^\la_i \G^i_\la \om
\ee
of the affine bundle  (\ref{b418'}). Consequently, it defines a
Hamiltonian form
\mar{ps14}\beq
h_\G^*\Xi=p^\la_i dy^i\w\om_\la -p^\la_i\G^i_\la\om \label{ps14}
\eeq
on a Legendre bundle $\Pi$ which coincides with the form $H_\G$
(\ref{3.6}).

(ii) Hamiltonian forms constitute a non-empty affine space
modelled over the linear space of densities $\wt H=\wt{\cH}\om$ on
$\Pi\to X$ which are sections of the pull-back bundle
(\ref{ps12}).

(iii) Given a connection $\G$ on $Y\to X$, every Hamiltonian form
$H$ (\ref{b418}) admits a decomposition
\mar{4.7}\beq
H=H_\G -\wt H_\G =p^\la_idy^i\w\om_\la
-p^\la_i\G^i_\la\om-\wt{\cH}_\G\om. \label{4.7}
\eeq
\end{theorem}

Moreover, every Hamiltonian form $H$ (\ref{b418}) admits the
canonical decomposition (\ref{3.8}) as follows.

We agree to call any bundle morphism
\mar{2.7}\beq
\Phi=dx^\la\otimes(\dr_\la
+\Phi^i_\la(x^\m,y^j,p^\m_j)\dr_i):\Pi\op\to_Y J^1Y \label{2.7}
\eeq
over $Y$ the  Hamiltonian map.

In particular, let $\G$ be a connection on $Y\to X$. Then the
composition
\mar{b420}\beq
\wh\G=\G\circ\pi_{\Pi Y}=dx^\la\otimes (\dr_\la
+\G^i_\la\dr_i):\Pi\to Y\to J^1Y \label{b420}
\eeq
is a Hamiltonian map. Conversely, every Hamiltonian map
$\Phi:\Pi\to J^1Y$ yields the associated connection
\mar{ps16}\beq
 \G_\Phi =\Phi\circ\wh 0=dx^\la\otimes(\dr_\la
+\Phi^i_\la(x^\m,y^j,0)\dr_i) \label{ps16}
\eeq
on $Y\to X$, where $\wh 0$ is the global zero section of a
Legendre bundle $\Pi\to Y$. In particular, we have
$\G_{\wh\G}=\G$.

\begin{theorem} \label{91t2} \mar{91t2} Every Hamiltonian map
(\ref{2.7}) defines a Hamiltonian form
\mar{414}\beq
H_\Phi=-\Phi\rfloor\bth_Y =p^\la_idy^i\w\om_\la
-p^\la_i\Phi^i_\la\om. \label{414}
\eeq
\end{theorem}

\begin{proof}
Given an arbitrary connection $\G$ on a fibre bundle $Y\to X$, the
corresponding Hamiltonian map (\ref{b420}) defines the form
$-\wh\G\rfloor\bth_Y$ which is exactly the Hamiltonian form $H_\G$
(\ref{3.6}). Since $\Phi-\wh\G$ is a $VY$-valued basic one-form on
$\Pi\to X$, $H_\Phi-H_\G$ is a density on $\Pi$. Then the result
follows from item (ii) of Theorem \ref{91t1}.
\end{proof}

The converse also is true.

\begin{theorem}\label{hammap}
Every Hamiltonian form $H$ (\ref{b418}) on a Legendre bundle
$\Pi\to Y$ yields the associated Hamiltonian map
\mar{415}\beq
\wh H:\Pi\to J^1Y,\qquad y_\la^i\circ\wh H=\dr^i_\la\cH.
\label{415}
\eeq
\end{theorem}

\begin{proof} In accordance with Theorem \ref{91t3} below, any
Hamiltonian form $H$ admits the associated Hamiltonian connection
$\g_H$ (\ref{cmp3}) and defines the Hamiltonian map (\ref{415}):
\mar{415'}\beq
\wh H=J^1\pi_{\Pi Y}\circ \g_H:\Pi \to J^1\Pi \to J^1Y.
\label{415'}
\eeq
\end{proof}

\begin{corollary}\label{hammap2} Every  Hamiltonian form $H$ (\ref{b418}) on
a Legendre bundle $\Pi\to Y$ yields the associated connection
(\ref{ps16}):
\mar{ps17}\beq
\G_H =\wh H\circ\wh 0=dx^\la\otimes(\dr_\la
+\dr^i_\la\cH(x^\m,y^j,0)\dr_i) \label{ps17}
\eeq
on a fibre bundle $Y\to X$.
\end{corollary}

In particular, we have $\G_{H_\G}=\G$, where $H_\G$ is the
Hamiltonian form (\ref{3.6}) associated to a connection $\G$ on
$Y\to X$.

\begin{corollary}
Every Hamiltonian form $H$ (\ref{b418}) admits the canonical
splitting
\beq
H=H_{\G_H}-\wt H.\label{3.8}
\eeq
\end{corollary}

\begin{remark} \mar{ps20} \label{ps20}
A Hamiltonian form is the main ingredient in PS Hamiltonian
formalism because just it defines the covariant Hamilton equation
(Section 6). Since Hamiltonian forms are the pull-back of the
canonical multisymplectic form (\ref{ps13}) on the homogeneous
Legendre bundle $Z_Y$ one can treat the latter as a homogeneous
momentum phase space of PS formalism by analogy with the cotangent
bundle in time-dependent mechanics (Section 7).
\end{remark}

\section{Covariant Hamilton equations}

Let $\g$ (\ref{cmp33}) be a Hamiltonian connection on a Legendre
bundle $\Pi\to X$.

Given a connection $\G$ on a fibre bundle $Y\to X$, the local form
$H_U$ (\ref{mos06}) in the expression (\ref{cmp4}) can be written
as
\be
H_U=H_\G - \wt\cH_\G\om,
\ee
where $H_\G$ is the Hamiltonian form (\ref{ps14}) and
$\wt\cH_\G\om$ is a local density on $\Pi$. In accordance with
item (ii) of Theorem  \ref{91t1}, it follows that $H_U$ is a local
Hamiltonian form. Thus, we have proved the following.

\begin{theorem}\label{lochamform}
For every Hamiltonian connection $\g$ (\ref{cmp33}) on a Legendre
bundle $\Pi\to X$, there exists a local Hamiltonian form $H_U$ in
a neighborhood $U$ of each point $q\in\Pi$ such that
\be
\g\rfloor\bom_Y =dH_U.
\ee
\end{theorem}

The converse is the following.

\begin{theorem} \label{91t3} \mar{91t3}
Every Hamiltonian form $H$ (\ref{b418}) admits a Hamiltonian
connection $\g_H$  which obeys the condition
\mar{cmp3}\beq
\g_H\rfloor\bom_Y= dH, \qquad \g^i_\la =\dr^i_\la\cH, \qquad
\g^\la_{\la i}=-\dr_i\cH. \label{cmp3}
\eeq
\end{theorem}

\begin{proof}
It is readily observed that the Hamiltonian form $H$ (\ref{b418})
is the Poincar\'e -- Cartan form (\ref{303}) of an affine first
order Lagrangian
\mar{Q3}\beq
L_H=h_0(H) = (p^\la_iy^i_\la - \cH)\om \label{Q3}
\eeq
on the jet manifold $J^1\Pi$. The Euler -- Lagrange operator
(\ref{305}) associated to this Lagrangian reads
\mar{3.9}\ben
&& \dl L_H=\cE_H :J^1\Pi\to T^*\Pi\w(\op\w^n T^*X),\nonumber \\
&& \cE_H=[(y^i_\la-\dr^i_\la\cH) dp^\la_i
-(p^\la_{\la i}+\dr_i\cH) dy^i]\w\om. \label{3.9}
\een
It is called the Hamilton operator for $H$. A glance at the
expression (\ref{3.9}) shows that this operator is an affine
morphism over $\Pi$ of constant rank. It follows that its kernel
\mar{b4100a,b}\ben
&& y^i_\la=\dr^i_\la\cH, \label{b4100a}\\
&& p^\la_{\la i}=-\dr_i\cH \label{b4100b}
\een
is an affine closed imbedded subbundle of the jet bundle
$J^1\Pi\to\Pi$. Therefore, it admits a global section $\g_H$ which
thus is a desired Hamiltonian connection obeying the relation
(\ref{cmp3}).
\end{proof}

\begin{remark} \label{91r10} \mar{91r10}
In fact, the Lagrangian (\ref{Q3}) is the pull-back onto $J^1\Pi$
of a form $L_H$ on the product $\Pi\op\times_Y J^1Y$.
\end{remark}

It should be emphasized that, if $\di X>1$, there is a set of
Hamiltonian connections associated to the same Hamiltonian form
$H$. They differ from each other in soldering forms $\si$ on
$\Pi\to X$ which fulfill the equation $\si\rfloor\bom_Y=0$.

\begin{example}
Any connection $\wt\G$ (\ref{404}) for different connections $K$
(\ref{08}) is a Hamiltonian connection for the Hamiltonian form
$H_\G$ (\ref{3.6}).
\end{example}

Being a closed imbedded subbundle of the jet bundle $J^1\Pi\to X$,
the kernel (\ref{b4100a})  --  (\ref{b4100b}) of the Euler --
Lagrange operator $\cE_H$ (\ref{3.9}) defines the Euler --
Lagrange equation on $\Pi$. It is a first order dynamic equation
called the covariant Hamilton equation.

Every integral section $r$ (i.e.,  $J^1r=\g_H\circ r$) of a
Hamiltonian connection $\g_H$ associated to a Hamiltonian form $H$
is obviously a solution of the covariant Hamilton equation
(\ref{b4100a})  --  (\ref{b4100b}). Conversely, if $r:X\to\Pi$ is
a global solution of the covariant Hamilton equation
(\ref{b4100a})  --  (\ref{b4100b}), there exists an extension of
the local section
\be
J^1r: r(X) \to \Ker\cE_H
\ee
to a Hamiltonian connection $\g_H$ which has $r$ as an integral
section. Substituting $J^1r$ in (\ref{415'}), we obtain the
equality
\mar{N10}\beq
J^1(\pi_{\Pi Y}\circ r)= \wh H\circ r \label{N10}
\eeq
for every solution $r$ of the covariant Hamilton equation
(\ref{b4100a}) -- (\ref{b4100b}) which thus is equivalent to this
equation.

\begin{remark}
Similarly to the Cartan equation (\ref{C28}), the covariant
Hamilton equation (\ref{b4100a})  --  (\ref{b4100b}) is equivalent
to the condition
\mar{N7}\beq
r^*(u\rfloor dH)= 0 \label{N7}
\eeq
for any vertical vector field $u$ on $\Pi\to X$.
\end{remark}

\begin{remark}
The covariant Hamilton equation (\ref{b4100a})  --  (\ref{b4100b})
has a standard form
\be
S^\la_{ab}(x, \phi)\dr_\la\phi^b = f_a(x,\phi)
\ee
for the Cauchy problem or, to be more precise, for the general
Cauchy problem since the coefficients $S^\la_{ab}$ depend on the
variable functions $\phi$ in general \cite{book,book95}. Here
$\phi^b$ is a compact notation for variables $r^i$ and $r_i^\la$.
However, the characteristic form
\be
\det(S^\la_{ab}c_\la), \qquad c_\la\in\mathbb R,
\ee
of this system fails to be different from zero for any $c_\la$.
One can overcome this difficulty as follows. Let us single out a
local coordinate $x^1$ and replace the equation (\ref{b4100a})
with the equation
\ben
&&\dr_1 r^i =\dr^i_1\cH, \qquad
d_1\dr_\la^i\cH =d_\la\dr_1^i\cH, \qquad \la\neq 1, \label{4.3}\\
 && d_\m= \dr_\m +\dr_\m r^i\dr_i +\dr_\m
r^\la_i\dr^i_\la.\nonumber
\een
The systems (\ref{4.3}) and (\ref{b4100b}) have the standard form
for the Cauchy problem with the initial conditions
\beq
r^i(x')=\phi^i(x'),\qquad
 r_i^\m(x')=\phi_i^\m (x'), \qquad
\dr_\la r^i =\dr^i_\la\cH, \qquad \la\neq 1,\label{4.4}
\eeq
on a local hypersurface $S$ of $X$  transversal to coordinate
lines $x^1$. If $r^i$ and $r^i_\la$ are solutions of the Cauchy
problem (of class $C^2$) for the equations (\ref{4.3}) and
(\ref{b4100b}) with the initial conditions (\ref{4.4}), they
satisfy the equation (\ref{b4100a}). Thus, in order to formulate
the Cauchy problem for a covariant Hamilton equation in PS
Hamiltonian formalism, one should single a one of the coordinates
out and consider the system of equations (\ref{4.3}) and
(\ref{b4100b}).
\end{remark}

\section{Hamiltonian time-dependent mechanics}

As was mentioned above, if $X=R=\mathbb R$, we are in a case of
time-dependent Hamiltonian mechanics which admits time
reparametrization $t\to t'(t)$ \cite{book10}.

Let $Q\to R$ be its configuration bundle coordinated by $(t,q^a)$.
The corresponding Legendre bundle (\ref{00}) is
\mar{xx4}\beq
\Pi=V^*Q\op\ot_Q T^*R\op\ot_Q TR. \label{xx4}
\eeq
It is endowed with holonomic coordinates $(t,q^a,p_a)$ possessing
transition functions
\be
p'_a=\frac{\dr q^b}{\dr q'a}p_b.
\ee
This transformation law is the same as that of fibre coordinates
on the vertical cotangent bundle $V^*Q\to Q$ of $Q\to R$.
Therefore, we have the canonical isomorphism $\Pi=V^*Q$
(\ref{000}) of the Legendre bundle $\Pi$ (\ref{xx4}) of
time-reparametrized mechanics to a momentum phase space $V^*Q$ of
time-dependent mechanics over the time axis $R=\mathbb R$ provided
with the canonical Cartesian coordinate $t$ with transition
functions $t'=t+$const. \cite{book10,sard98,sard13}. Accordingly,
the homogeneous Legendre bundle $Z_Q$ (\ref{N41}) of
time-reparametrized mechanics is isomorphic to the homogeneous
momentum phase space $T^*Q$ of time-dependent mechanics over the
time axis $\mathbb R$. It is endowed  with the holonomic
coordinates $(t,q^a,p_0,p_a)$, possessing transition functions
\mar{2.3s}\beq
{p'}_a = \frac{\dr q^b}{\dr{q'}^a}p_b, \qquad p'_0=
\left(p_0+\frac{\dr q^b}{\dr t}p_b\right). \label{2.3s}
\eeq

\begin{remark} \mar{ps21} \label{ps21}
Note that, relative to the Cartesian coordinate $t$, the time axis
$\mathbb R$ is provided with the standard vector field $\dr_t$ and
the standard one-form $dt$ which also is the volume element on
$\mathbb R$. As a consequence, there is the one-to-one
correspondence between the vector fields $f\dr_t$, the densities
$fdt$ and the real functions $f$ on $\mathbb R$. Roughly speaking,
one can neglect the contribution of $T\mathbb R$ and $T^*\mathbb
R$ to some expressions. In particular, the canonical imbedding
(\ref{18}) of $J^1Q$ takes a form
\mar{z260}\ben
&& \la_{(1)}: J^1Q\ni (t,q^a,q^a_t)\to (t,q^a,\dot t=1, \dot q^a=q^a_t) \in TQ, \label{z260}\\
&& \la_{(1)}=d_t=\dr_t +q^a_t\dr_a. \nonumber
\een
In view of the morphism $\la_{(1)}$ (\ref{z260}), any connection
\mar{z270}\beq
\G=dt\ot (\dr_t +\G^a\dr_a) \label{z270}
\eeq
on a fibre bundle $Q\to\mathbb R$  can be identified with a
nowhere vanishing horizontal vector field
\mar{a1.10}\beq
\G = \dr_t + \G^a \dr_a \label{a1.10}
\eeq
on $Q$ which is the horizontal lift $\G\dr_t$ of the standard
vector field $\dr_t$ on $\mathbb R$ by means of the connection
(\ref{z270}). Conversely, any vector field $\G$ on $Q$ such that
$dt\rfloor\G =1$ defines a connection on $Q\to\mathbb R$.
Therefore, the connections (\ref{z270}) conventionally are
identified with the vector fields (\ref{a1.10}). The integral
curves of the vector field (\ref{a1.10}) coincide with the
integral sections for the connection (\ref{z270}).
\end{remark}

A homogeneous momentum phase space $T^*Q$ admits the canonical
Liouville form
\mar{N43s}\beq
\Xi_T =p_odt +p_adq^a \label{N43s}
\eeq
(cf. (\ref{N43})) and the canonical symplectic form
\mar{m91'}\beq
\Om_T=d\Xi=dp_0\w dt +dp_a\w dq^a \label{m91'}
\eeq
(cf. (\ref{ps13})) together with the corresponding Poisson bracket
\mar{m116}\beq
\{f,g\}_T =\dr^0f\dr_tg - \dr^0g\dr_tf +\dr^af\dr_ag-\dr^ag\dr_af,
\quad f,g\in C^\infty(T^*Q). \label{m116}
\eeq
This bracket yields the coinduced Poisson bracket (\ref{m72}) on a
momentum phase space $V^*Q$ as follows.

There is the canonical one-dimensional affine bundle
\mar{z11'}\beq
\zeta:T^*Q\to V^*Q \label{z11'}
\eeq
(cf. (\ref{b418'})). A glance at the transformation law
(\ref{2.3s}) shows that it is a trivial affine bundle. Indeed,
given a global section $h$ of $\zeta$, one can equip $T^*Q$ with a
global fibre coordinate
\be
I_0=p_0-h, \qquad I_0\circ h=0,
\ee
possessing the identity transition functions. With respect to the
coordinates $(t,q^a,I_0,p_a)$ the fibration (\ref{z11'}) reads
\mar{z11}\beq
\zeta: \mathbb R\times V^*Q \ni (t,q^a,I_0,p_a)\to (t,q^a,p_a)\in
V^*Q. \label{z11}
\eeq

Let us consider the subring of $C^\infty(T^*Q)$ which comprises
the pull-back $\zeta^*f$ onto $T^*Q$ of functions $f$ on the
vertical cotangent bundle $V^*Q$ by the fibration $\zeta$
(\ref{z11'}).  This subring is closed under the Poisson bracket
(\ref{m116}). Then by virtue of the well-known theorem
\cite{book10,vais}, there exists the degenerate coinduced Poisson
structure
\mar{m72}\beq
\{f,g\}_V = \dr^af\dr_ag-\dr^ag\dr_af, \qquad f,g\in
C^\infty(V^*Q), \label{m72}
\eeq
on a the vertical cotangent bundle $V^*Q$ such that
\mar{m72'}\beq
\zeta^*\{f,g\}_V=\{\zeta^*f,\zeta^*g\}_T.\label{m72'}
\eeq
The holonomic coordinates $(t,q^a,p_a)$ on $V^*Q$ are canonical
for the Poisson structure (\ref{m72}).

With respect to the vertical Poisson bracket (\ref{m72}), the
Hamiltonian vector fields of functions on $V^*Q$ read
\mar{m73}\beq
\vt_f = \dr^if\dr_i- \dr_if\dr^i, \qquad
[\vt_f,\vt_{f'}]=\vt_{\{f,f'\}_V}, \qquad f,f'\in C^\infty(V^*Q).
\label{m73}
\eeq
They are vertical vector fields on $V^*Q\to \mathbb R$.
Accordingly, the characteristic distribution of the Poisson
structure (\ref{m72}) is the vertical tangent bundle $VV^*Q\subset
TV^*Q$ of a fibre bundle $V^*Q\to \mathbb R$. The corresponding
symplectic foliation on the momentum phase space $V^*Q$ coincides
with a fibration $V^*Q\to \mathbb R$.

One can think of the vertical Poisson bracket (\ref{m72}) as being
the particular PS bracket (\ref{xx3}). Indeed, in view of Remark
\ref{ps21}, the algebra of multivector densities (\ref{ps11})
comes to the ring $C^\infty(V^*Q)$, whereas the tangent-valued
Liouville form (\ref{2.4}) and the PS form (\ref{406}) are
associated to the pull-back forms
\mar{z401}\ben
&& \bth=h^*(\Xi_T\w dt)=p_idq^i\w dt, \nonumber\\
&& \bom=h^*(\Om_T\w dt)=dp_i\w dq^i\w dt \label{z401}
\een
on  $V^*Q$, where $h$ is some section of the trivial bundle
(\ref{z11'}). They are independent of the choice of $h$.  With
$\bom$ (\ref{z401}), the Hamiltonian vector field $\vt_f$
(\ref{m73}) for a function $f$ on $V^*Q$ is given by the relation
\be
\vt_f\rfloor\bom = -df\w dt,
\ee
while the vertical Poisson bracket (\ref{m72}) is written as
\be
\{f,g\}_Vdt=\vt_g\rfloor\vt_f\rfloor\bom
\ee
similarly to the PS bracket (\ref{xx3}).

In contrast with autonomous Hamiltonian mechanics, the Poisson
structure (\ref{m72}) fails to provide any dynamic equation on a
fibre bundle $V^*Q\to\mathbb R$ because Hamiltonian vector fields
(\ref{m73}) of functions on $V^*Q$ are vertical vector fields, but
not connections on $V^*Q\to\mathbb R$. Hamiltonian dynamics on
$V^*Q$ is described as a particular PS Hamiltonian dynamics in
Section 5 \cite{book10,sard13}.

A Hamiltonian on a momentum phase space $V^*Q\to\mathbb R$ of
non-relativistic mechanics is defined as a global section
\mar{ws513}\beq
h:V^*Q\to T^*Q, \qquad p_0\circ h=\cH(t,q^j,p_j), \label{ws513}
\eeq
of the affine bundle $\zeta$ (\ref{z11'}). Given the Liouville
form $\Xi_T$ (\ref{N43s}) on $T^*Q$, this section yields the
pull-back Hamiltonian form
\mar{b4210s}\beq
H=(-h)^*\Xi= p_k dq^k -\cH dt  \label{b4210s}
\eeq
on $V^*Q$. This is the well-known invariant of Poincar\'e --
Cartan \cite{arn}.

Given a Hamiltonian form $H$ (\ref{b4210s}), there exists a unique
horizontal vector field (\ref{a1.10}):
\be
\g_H=\dr_t -\g^i\dr_i -\g_i\dr^i,
\ee
on $V^*Q$ (i.e., a connection on $V^*Q\to \mathbb R$) such that
\mar{w255}\beq
\g_H\rfloor dH=0. \label{w255}
\eeq
This vector field, called the Hamilton vector field, reads
\mar{z3}\beq
\g_H=\dr_t + \dr^k\cH\dr_k- \dr_k\cH\dr^k. \label{z3}
\eeq
In a different way, a Hamilton vector field $\g_H$ is defined by
the relation
\be
\g_H\rfloor\bom=dH
\ee
(cf. (\ref{cmp3})).  This vector field yields the first order
dynamic Hamilton equation
\mar{z20}\beq
q^k_t=\dr^k\cH, \qquad  p_{tk}=-\dr_k\cH \label{z20}
\eeq
on $V^*Q\to\mathbb R$, where $(t,q^k,p_k,q^k_t,\dot p_{tk})$ are
the adapted coordinates on the first order jet manifold $J^1V^*Q$
of $V^*Q\to\mathbb R$. A solution of the Hamilton equation
(\ref{z20}) is an integral section for the connection $\g_H$
(\ref{z3}).

\section{Iso-PS structure}

The canonical PS structure defined by the tangent-valued Liouville
form $\bth_Y $ (\ref{2.4}) and the PS form $\bom_Y$ (\ref{406}) is
by no means the unique PS structure on the Legendre bundle $\Pi$
(\ref{00}). One can consider its following deformation
\cite{book95}.

Let
\be
\psi =\psi^\la_\mu(x)dx^\mu\otimes\dr_\la
\ee
be a tangent-valued one-form on $X$ corresponding to some
isomorphism of the tangent bundle $TX$ of $X$.  Given the the
tangent-valued Liouville form $\bth_Y $ (\ref{2.4}) and the PS
form $\bom_Y$ (\ref{406}), let us consider their deformations
\ben
&& \bth_\psi=\bth_Y\rfloor\psi= \psi^\la_\mu (x)p^\mu_idy^i\w
\om\ot\dr_\la, \label{M9} \\
&&\bom_\psi=\bom_Y\rfloor\psi= \psi^\la_\mu(x)dp^\mu_i\w
dy^i\w\om\otimes \dr_\la, \label{M1}
\een
respectively. In comparison with the canonical forms $\bth_Y$ and
$\bom_Y$, the forms $\bth_\psi$ (\ref{M9}) and $\bom_\psi$
(\ref{M1}) provide another PS structure on a Legendre bundle
$\Pi$. We agree to call call it the iso-PS structure. In
particular, the relation
\beq
\bom_\psi\rfloor\f =d(\bth_\psi\rfloor\f) \label{ps29}
\eeq
(cf. (\ref{ps30})) holds for an arbitrary exterior one-form $\f$
on $X$.

Building on the forms $\bth_\psi$ and $\bom_\psi$, one can develop
iso-PS Hamiltonian formalism by analogy with the canonical PS one
in Sections 3 -- 5.

Let us say that that the connection $\g$ (\ref{cmp33}) on $\Pi\to
X$ is an iso-Hamiltonian connection if an exterior form
$\g\rfloor\bom_\psi$ is closed.

An exterior $n$-form $H_\psi$ on $\Pi$ is called the
iso-Hamiltonian form if there exists an iso-Hamiltonian connection
$\g$ such that
\be
\g\rfloor\bom_\psi =dH_\phi.
\ee

The following assertion shows that sets of iso-Hamiltonian
connections and iso-Hamiltonian forms on $\Pi$ are not empty.

\begin{theorem} \label{3.32} Let $\G$ be the connection (\ref{ps9}) on $Y\to X$ and
$\wt\G$ (\ref{404}) its lift onto $\Pi\to X$. We have an
iso-Hamiltonian form
\be
H_{\G\psi} = \psi^\la_\mu(p^\mu_idy^i\w\om_\la - \G^i_\la
p^\la_i\om)
\ee
and the associated iso-Hamiltonian connection
\be
\g=\wt\G +\frac 1n (\psi^{-1})^\mu_\la (-\dr_\al\psi^\al_\nu -
K^\al{}_{\beta\al}\psi^\beta_\nu +K^\al_{\nu\beta}\psi^\beta_\al)
p^\nu_idx^\la\otimes\dr^i_\mu.
\ee
\end{theorem}

Item (ii) of Theorem \ref{91t1} also is extended to
iso-Hamiltonian forms. Then as an immediate consequence of Theorem
\ref{3.32}, we obtain the following corollary.

\begin{corollary} Any iso-Hamiltonian form is given by the
expression
\beq
H_\psi=H_{\G\psi} + \wt{\cH}_\G\om =\psi^\la_\mu
p^\mu_idy^i\w\om_\la -\cH_\psi\om \label{M12}
\eeq
where $\G$ is a connection on $Y\to X$.
\end{corollary}

By analogy with PS Hamiltonian formalism, one can introduce the
Hamilton operator and obtain the covariant Hamilton equation
associated to the iso-Hamiltonian form (\ref{M12}). For sections
$r$ of a Legendre bundle $\Pi\to X$, this equation reads
\be
\psi^\la_\mu\dr_\la r^i =\dr_\mu^i\cH_\psi, \qquad \dr_\la
(\psi^\la_\mu r^\mu_i)= -\dr_i\cH_\psi.
\ee

\section{Associated Hamiltonian and Lagrangian systems}

Let us study the relations between first order Lagrangian and
covariant Hamiltonian formalisms \cite{book,jpa99,book10}.

We are based on the fact that any Lagrangian $L$ on a
configurations space $J^1Y$ defines the Legendre map $\wh L$
(\ref{b330}) and its jet prolongation
\be
&& J^1\wh L: J^!J^1Y\ar_Y J^1\Pi, \qquad
(p^\la_i, y^i_\m,p^\la_{\m i})\circ J^1\wh L= (\pi^\la_i,\wh
y^i_\m, \wh d_\m\pi^\la_i),\\
&& \wh d_\la=\dr_\la+ \wh y^j_\la\dr_j+y^j_{\la\m}\dr_j^\m,
\ee
and that any Hamiltonian form $H$ on a momentum phase space $\Pi$
yields the Hamiltonian map $\wh H$ (\ref{415}) and its jet
prolongation
\be
&& J^1\wh H: J^1\Pi\ar_Y J^1J^1Y, \qquad
(y^i_\m,\wh y^i_\la,y^i_{\la\m})\circ J^1\wh
H=(\dr^i_\m\cH,y^i_\la,d_\la\dr^i_\m\cH),\\
&& d_\la=\dr_\la+y^j_\la\dr_j+p^\nu_{\la j}\dr^j_\nu.
\ee

Let us start with the case of a hyperregular Lagrangian $L$, i.e.,
when the corresponding Legendre map $\wh L$ is a diffeomorphism.
Then $\wh L^{-1}$ is a Hamiltonian map. Let us consider the
Hamiltonian form
\mar{cc311}\beq
H=H_{\wh L^{-1}}+\wh L^{-1*}L,  \qquad \cH=p^\m_i\wh L^{-1}{}_\m^i
- \cL(x^\m, y^j,\wh L^{-1}{}_\m^j), \label{cc311}
\eeq
where $H_{\wh L^{-1}}$ is the Hamiltonian form (\ref{414})
associated to the Hamiltonian map $\wh L^{-1}$. Let $s$ be a
solution of the Euler -- Lagrange equation (\ref{b327}) for a
Lagrangian $L$. A direct computation shows that $\wh L\circ J^1s$
is a solution of the covariant Hamilton equation (\ref{b4100a}) --
(\ref{b4100b}) for the Hamiltonian form $H$ (\ref{cc311}).
Conversely, if $r$ is a solution of the covariant Hamilton
equation (\ref{b4100a}) -- (\ref{b4100b}) for the Hamiltonian form
$H$ (\ref{cc311}), then $s=\pi_{\Pi Y}\circ r$ is a solution of
the Euler -- Lagrange equation (\ref{b327}) for $L$ (see the
equality (\ref{N10})). Thus, one can state the following.

\begin{theorem} \mar{ps25} \label{ps25}
In a case of hyperregular Lagrangians, Lagrangian and PS
Hamiltonian formalisms are equivalent.
\end{theorem}

Let now $L$ be an arbitrary Lagrangian on a configuration space
$J^1Y$.

\begin{definition}
A Hamiltonian form $H$ is said to be associated to a Lagrangian
$L$ if $H$ satisfies the relations
\mar{2.30}\ben
&& \wh L\circ\wh H\circ \wh L=\wh L,\label{2.30a} \\
&& H=H_{\wh H}+\wh H^*L. \label{2.30b}
\een
\end{definition}

A glance at the relation (\ref{2.30a}) shows that $\wh L\circ\wh
H$ is a projector
\mar{b481'}\beq
p^\m_i(p)=\dr^\m_i\cL (x^\m,y^i,\dr^j_\la\cH(p)), \qquad p\in N_L,
\label{b481'}
\eeq
from $\Pi$ onto the Lagrangian constraint space $N_L=\wh L( J^1Y)$
(\ref{ps3}). Accordingly,  $\wh H\circ\wh L$ is a projector from
$J^1Y$ onto $\wh H(N_L)$.

\begin{definition}
A Hamiltonian form is called weakly associated to a Lagrangian $L$
if the condition (\ref{2.30b}) holds on a Lagrangian constraint
space $N_L$. \end{definition}

\begin{theorem} \label{jp} \mar{jp}
If a Hamiltonian map $\Phi$ (\ref{2.7}) obeys a relation
\mar{ps130}\beq
\wh L\circ\Phi \circ \wh L=\wh L, \label{ps130}
\eeq
then a Hamiltonian form
\mar{ps131}\beq
H=H_\Phi+\Phi^*L \label{ps131}
\eeq
is weakly associated to a Lagrangian $L$. If $\Phi=\wh H$, then
$H$ is associated to $L$.
\end{theorem}

\begin{theorem} \label{cmp110} \mar{cmp110}
Any Hamiltonian form $H$ weakly associated to a Lagrangian $L$
fulfills a relation
\mar{4.9}\beq
H|_{N_L}=\wh H^*H_L|_{N_L}, \label{4.9}
\eeq
where $H_L$ is the Poincar\'e -- Cartan form (\ref{303}).
\end{theorem}

\begin{proof}
The relation (\ref{2.30b}) takes a coordinate form
\mar{b481}\beq
\cH(p)=p^\m_i\dr^i_\m\cH-\cL(x^\m,y^i,\dr^j_\la\cH(p)), \qquad
p\in N_L. \label{b481}
\eeq
Substituting (\ref{b481'}) and (\ref{b481}) in (\ref{b418}), we
obtain the relation (\ref{4.9}).
\end{proof}

The difference between associated and weakly associated
Hamiltonian forms lies in the following. Let $H$ be an associated
Hamiltonian form, i.e., the equality (\ref{b481}) holds everywhere
on $\Pi$. Acting on this equality by the exterior differential, we
obtain the relations
\mar{2.31,'}\ben
&& \dr_\m\cH(p) =-(\dr_\m\cL)\circ \wh H(p), \qquad p\in N_L,
\nonumber\\
&&  \dr_i\cH(p) =-(\dr_i\cL)\circ \wh H(p), \qquad
p\in N_L, \label{2.31}\\
&&
(p^\m_i-(\dr^\m_i\cL)
(x^\m,y^i,\dr^j_\la\cH))\dr^i_\m\dr^a_\al\cH=0.\label{2.31'}
\een
The relation (\ref{2.31'}) shows that the associated Hamiltonian
form (i.e., the Hamiltonian map $\wh H$) is not regular outside a
Lagrangian constraint space $N_L$.

\begin{example} \mar{ps22} \label{ps22}
Any Hamiltonian form is weakly associated to a zero Lagrangian
$L=0$, while the associated one is only $H_\G$ (\ref{3.6}).
\end{example}

\begin{example}
A hyperregular Lagrangian has a unique associated and weakly
associated Hamiltonian form (\ref{cc311}). In a case of a regular
Lagrangian $L$, the Lagrangian constraint space $N_L$ is an open
subbundle of a vector Legendre bundle $\Pi\to Y$. If $N_L\neq\Pi$,
a weakly associated Hamiltonian form fails to be defined
everywhere on $\Pi$ in general. At the same time, $N_L$ itself can
be provided with the pull-back PS structure with respect to the
imbedding $N_L\to \Pi$, so that one may consider Hamiltonian forms
on $N_L$.
\end{example}

One can say something more in a case of semiregular Lagrangians
(Definition \ref{d11}).

\begin{lemma} \label{3.22} \mar{3.22}
The Poincar\'e -- Cartan form $H_L$ for a semiregular Lagrangian
$L$ is constant on the connected inverse image $\wh L^{-1}(p)$ of
any point  $p\in N_L$.
\end{lemma}

\begin{proof}
Let $u$ be a vertical vector field on an affine jet bundle
$J^1Y\to Y$ which takes its values into the kernel of the tangent
map $T\wh L$ to $\wh L$. Then the Lie derivative $\bL_u H_L$ of
$H_L$ along $u$ vanishes.
\end{proof}

A corollary of Lemma \ref{3.22} is the following.

\begin{theorem} \label{3.22'} \mar{3.22'} All Hamiltonian forms weakly associated
to a semiregular Lagrangian $L$ coincide with each other on a
Lagrangian constraint space $N_L$, and the Poincar\'e -- Cartan
form $H_L$ (\ref{303}) for $L$ is the pull-back
\mar{2.32}\beq
H_L=\wh L^*H, \qquad
(\pi^\la_iy^i_\la-\cL)\om=\cH(x^\m,y^j,\pi^\m_j)\om,\label{2.32}
\eeq
of any such a Hamiltonian form $H$.
\end{theorem}

\begin{proof}
Given a vector $v\in T_p\Pi$, the value $T\wh H(v)\rfloor H_L(\wh
H(p))$ is the same for all Hamiltonian maps $\wh H$ satisfying the
relation (\ref{2.30a}). Then the result follows from the relation
(\ref{4.9}).
\end{proof}

Theorem \ref{3.22'} enables us to relate the Euler -- Lagrange
equation for a semiregular Lagrangian $L$ with the covariant
Hamilton equations for Hamiltonian forms weakly associated to $L$
\cite{book,sardz93,book95}.

\begin{theorem}\label{3.23} \mar{3.23}
Let a section $r$ of $\Pi\to X$ be a  solution of the covariant
Hamilton equation (\ref{b4100a}) -- (\ref{b4100b}) for a
Hamiltonian form $H$ weakly associated to a semiregular Lagrangian
$L$. If $r$ lives in a Lagrangian constraint space $N_L$, a
section $s=\pi_{\Pi Y}\circ r$ of $Y\to X$ satisfies the Euler --
Lagrange equation (\ref{b327}), while its first order jet
prolongation
\be
\ol s=\wh H\circ r =J^1s
\ee
obeys the Cartan equation (\ref{b336a}) -- (\ref{b336b}).
\end{theorem}

\begin{proof}
Put $\ol s=\wh H\circ r$. Since $r(X)\subset N_L$, then
\be
r=\wh L\circ \ol s, \qquad J^1r=J^1\wh L\circ J^1\ol s.
\ee
If $r$ is a solution of the covariant Hamilton equation, the
exterior form $\cE_H$ vanishes at points of $J^1r(X)$. Hence, the
pull-back form $\cE_{\ol L}=(J^1\wh L)^*\cE_H$ vanishes at points
$J^1\ol s(X)$. It follows that the section $\ol s$ of a jet bundle
$J^1Y\to X$ obeys the Cartan equation  (\ref{b336a}) --
(\ref{b336b}). By virtue of the relation (\ref{N10}), we have $\ol
s= J^1s$. Hence, $s$ is a solution of the Euler -- Lagrange
equation.
\end{proof}

The converse assertion is more intricate (\cite{book}, Proposition
4.5.11).

\begin{theorem}\label{3.24s} \mar{3.24s} Given a semiregular Lagrangian $L$, let
a section $\ol s$ of the jet bundle $J^1Y\to X$ be a solution of
the Cartan equation (\ref{b336a}) -- (\ref{b336b}). Let $H$ be a
Hamiltonian form weakly associated to $L$ so that the associated
Hamiltonian map satisfies a condition
\mar{2.36'}\beq
\wh H\circ \wh L\circ \ol s=J^1(\pi^1_0\circ\ol s). \label{2.36'}
\eeq
Then, a section
\be
r=\wh L\circ \ol s,\qquad r^\la_i =\pi^\la_i(x^\la,\ol s^j,\ol
s^j_\la), \qquad r^i=\ol s^i,
\ee
of a Legendre bundle $\Pi\to X$ is a solution of the Hamilton
equation (\ref{b4100a}) -- (\ref{b4100b}) for $H$.
\end{theorem}

Being restricted to solutions of Euler -- Lagrange equations,
Theorem \ref{3.24s} comes to the following.

\begin{theorem}\label{3.24} \mar{3.24} Given a semiregular Lagrangian $L$,
let a section $s$ of a fibre bundle $Y\to X$ be a solution of the
Euler -- Lagrange equation (\ref{b327}) (i.e., $J^1s$ is a
solution of the Cartan equation (\ref{b336a}) -- (\ref{b336b}),
and $s=\pi^1_0\circ J^1s$). Let $H$ be a Hamiltonian form weakly
associated to $L$, and let $H$ satisfy a relation
\mar{2.36}\beq
\wh H\circ \wh L\circ J^1s=J^1s.\label{2.36}
\eeq
Then a section $r=\wh L\circ J^1s$ of a fibre bundle $\Pi\to X$ is
a solution of the covariant Hamilton equation (\ref{b4100a}) --
(\ref{b4100b}) for $H$.
\end{theorem}

\begin{example} \label{lagham10}
 Let $L=0$.
This Lagrangian is semiregular. Its Euler -- Lagrange equation
comes to the identity $0=0$. Every section $s$ of a fibre bundle
$Y\to X$ is a solution of this equation. Given a section $s$, let
$\G$ be a connection on $Y$ such that $s$ is its integral section.
The Hamiltonian form $H_\G$ (\ref{3.6}) is associated to $L=0$,
and the Hamiltonian map $\wh H_\G$ satisfies the relation
(\ref{2.36}). The corresponding Hamilton equation has a solution
\be
r=\wh L\circ J^1s, \qquad r^i=s^i, \qquad r^\la_i=0.
\ee
\end{example}

In view of Theorem \ref{3.24}, one may try to consider a set of
Hamiltonian forms associated to a semiregular Lagrangian $L$  in
order to exhaust all solutions of the Euler -- Lagrange equation
for $L$.

\begin{definition}
Let us say that a set of Hamiltonian forms $H$ weakly associated
to a semiregular Lagrangian $L$ is complete if, for each solution
$s$ of the Euler -- Lagrange equation for $L$, there exists a
solution $r$ of the covariant Hamilton equation for a Hamiltonian
form $H$ from this set such that $s=\pi_{\Pi Y}\circ r$.
\end{definition}

A complete family of Hamiltonian forms associated to a given
Lagrangian need not exist, or it fails to be defined uniquely. For
instance, Example \ref{lagham10} shows that the Hamiltonian forms
(\ref{3.6}) constitute a complete family associated to the zero
Lagrangian, but this family is not minimal.

By virtue of Theorem \ref{3.24}, a set of weakly associated
Hamiltonian forms is complete if, for every solution $s$ of the
Euler -- Lagrange equation for $L$, there is a Hamiltonian form
$H$ from this set which fulfills the relation (\ref{2.36}).

In a case of almost regular Lagrangians (Definition \ref{d11}),
one can formulate the following necessary and sufficient
conditions of the existence of weakly associated Hamiltonian
forms. An immediate consequence of Theorem \ref{jp} is the
following.

\begin{theorem} \label{mos08} \mar{mos08}
A Hamiltonian form $H$ weakly associated to an almost regular
Lagrangian $L$ exists iff the fibred manifold $J^1Y\to N_L$
(\ref{cmp12}) admits a global section.
\end{theorem}

\begin{proof}
A global section of $J^1Y\to N_L$ can be extended to a Hamiltonian
map $\Phi:\Pi\to J^1Y$ which obeys the relation (\ref{ps130}).
\end{proof}

In particular, on an open neighborhood $U\subset\Pi$ of each point
$p\in N_L\subset\Pi$, there exists a complete set of local
Hamiltonian forms weakly associated to an almost regular
Lagrangian $L$. Moreover, one can always construct a complete set
of local Hamiltonian forms associated to $L$ (\cite{book},
Proposition 4.5.14). At the same time, a complete set of
associated Hamiltonian forms may exists when a Lagrangian is not
necessarily semiregular (\cite{book}, Example 4.5.12).

Given a global section $\Psi$ of a fibred manifold
\mar{cmp12'}\beq
\wh L: J^1Y\to N_L, \label{cmp12'}
\eeq
let us consider the pull-back form
\mar{b4300}\beq
H_N=\Psi^*H_L=i_N^*H \label{b4300}
\eeq
on $N_L$ called the constrained Hamiltonian form.  By virtue of
Lemma \ref{3.22}, it does not depend on the choice of a section of
the fibred manifold (\ref{cmp12'}) and, consequently, $H_L=\wh
L^*H_N$. For sections $r$ of a fibre bundle $N_L\to X$, one can
write the constrained Hamilton equation
\mar{N44}\beq
r^*(u_N\rfloor dH_N) =0, \label{N44}
\eeq
where $u_N$ is an arbitrary vertical vector field on $N_L\to X$.
This equation possesses the following important properties.

\begin{theorem} \label{cmp22} \mar{cmp22} For any Hamiltonian form $H$ weakly
associated to an almost regular Lagrangian $L$, every solution $r$
of the covariant Hamilton equation which lives in a Lagrangian
constraint space $N_L$ is a solution of the constrained Hamilton
equation (\ref{N44}).
\end{theorem}

\begin{proof} Such a Hamiltonian form $H$ defines
a global section $\Psi=\wh H\circ i_N$ of the fibred manifold
(\ref{cmp12'}). Since $H_N=i^*_NH$ due to the relation
(\ref{2.32}), the constrained Hamilton equation can be written as
\mar{N44'}\beq
r^*(u_N\rfloor di^*_NH)=r^*(u_N\rfloor dH|_{N_L}) =0. \label{N44'}
\eeq
Note that this equation differs from the Hamilton equation
(\ref{N7}) restricted to $N_L$. This reads
\mar{cmp10}\beq
r^*(u\rfloor dH|_{N_L}) =0, \label{cmp10}
\eeq
where $r$ is a section of $N_L\to X$ and $u$ is an arbitrary
vertical vector field on $\Pi\to X$. A solution $r$ of the
equation (\ref{cmp10}) obviously satisfies the weaker condition
(\ref{N44'}).
\end{proof}

\begin{theorem}\label{3.02} \mar{3.02}
The constrained Hamilton equation (\ref{N44}) is equivalent to the
Hamilton -- De Donder equation (\ref{N46}).
\end{theorem}

\begin{proof}
It is readily observed that
\be
\wh L=\pi_{Z\Pi}\circ \wh H_L.
\ee
Hence, the projection $\pi_{Z\Pi}$ (\ref{b418'}) yields a
surjection of $Z_L$ onto $N_L$. Given a section $\Psi$ of the
fibred manifold (\ref{cmp12'}), we have a morphism
\be
\wh H_L\circ \Psi: N_L\to Z_L.
\ee
By virtue of Lemma (\ref{3.22}), this is a surjection such that
\be
\pi_{Z\Pi}\circ\wh H_L\circ \Psi=\id N_L.
\ee
Hence, $\wh H_L\circ \Psi$ is a bundle isomorphism over $Y$ which
is independent of the choice of a global section $\Psi$. Combining
(\ref{cmp14}) and (\ref{b4300}) results in
\be
H_N=(\wh H_L\circ \Psi)^*\Xi_L
\ee
that leads to a desired equivalence.
\end{proof}

This proof gives something more. Namely, since $Z_L$ and $N_L$ are
isomorphic, the homogeneous Legendre map $\wh H_L$ (\ref{N41})
fulfils the conditions of Theorem \ref{ddd}. Then combining
Theorem \ref{ddd} and Theorem \ref{3.02}, we obtain the following.

\begin{theorem}\label{3.01} \mar{3.01} Let $L$ be an almost regular Lagrangian
such that the fibred manifold (\ref{cmp12'}) has a global section.
A section $\ol s$ of the jet bundle $J^1Y\to X$ is a solution of
the Cartan equation (\ref{C28}) iff $\wh L\circ \ol s$ is a
solution of  the constrained Hamilton equation (\ref{N44}).
\end{theorem}

Theorem \ref{3.01} also is a corollary of Lemma \ref{cmp84} below.
The constrained Hamiltonian form $H_N$ (\ref{b4300}) defines the
constrained Lagrangian
\mar{cmp81}\beq
L_N=h_0(H_N)=(J^1i_N)^*L_H \label{cmp81}
\eeq
on the first order jet manifold $J^1N_L$ of a fibre bundle $N_L\to
X$.

\begin{lemma} \label{cmp84} \mar{cmp84}
There are the relations
\mar{cmp85}\beq
\ol L=(J^1\wh L)^*L_N, \qquad L_N=(J^1\Psi)^*\ol L, \label{cmp85}
\eeq
where $\ol L$ is the first order Lagrangian (\ref{cmp80}).
\end{lemma}

The Euler -- Lagrange equation for the constrained Lagrangian
$L_N$ (\ref{cmp81}) is equivalent to the constrained Hamilton
equation (\ref{N44}) and, by virtue of Lemma \ref{cmp84}, is
quasi-equivalent to the Cartan equation. At the same time, the
Cartan equation of a non-regular Lagrangian system may contain an
additional freedom in comparison with the constrained Hamilton
equation (Section 12).

\section{Lagrangian and Hamiltonian conservation laws}

In order to study symmetries of PS Hamiltonian theory, let us use
the fact that the Hamiltonian form $H$ (\ref{b418}) is a
Poincar\'e -- Cartan form for the Lagrangian $L_H$ (\ref{Q3}) and
that the covariant Hamilton equation for $H$ is the Euler --
Lagrange equation for $L_H$.

 We restrict our consideration to
classical symmetries defined by projectable vector fields
\mar{ps40}\beq
u=u^\la\dr_\la + u^i\dr_i \label{ps40}
\eeq
on a fibre bundle $Y\to X$.

Let us start with Lagrangian conservation laws in first order
Lagrangian formalism on a fibre bundle $Y\to X$
\cite{cmp05,book09,book13}.

The vector field $u$ (\ref{ps40}) admits the canonical
decomposition into the horizontal and vertical parts
\mar{23f40}\beq
u=u_H + u_V= (u^\la\dr_\la + y^i_\la\dr^\la_i) + (u^i\dr_i -
y^i_\la\dr^\la_i) \label{23f40}
\eeq
over $J^1Y$ and the first order jet prolongation
\mar{23f41}\beq
J^1u= u^\la\dr_\la + u^i\dr_i + (d_\la u^i - y^i_\m d_\la
u^\m)\dr^\la_i  \label{23f41}
\eeq
onto $J^1Y$.

Given the first order Lagrangian $L$ (\ref{23f2}), the global
variational decomposition (\ref{+421}) leads to the corresponding
splitting of the Lie derivative $\bL_{J^1u}L$ of $L$ along $J^1u$
(\ref{23f41}):
\mar{23f42}\ben
&& \bL_{J^1u}L= u_V\rfloor\cE_L + d_H(h_0(u\rfloor H_L)),
\label{23f42}\\
&& \dr_\la u^\la\cL +[u^\la\dr_\la+
u^i\dr_i +(d_\la u^i -y^i_\m\dr_\la u^\m)\dr^\la_i]\cL =  \nonumber\\
&& \qquad (u^i-y^i_\la u^\la)\cE_i
- d_\la[\pi^\la_i(u^\m y^i_\m -u^i) -u^\la\cL], \nonumber
\een
where $\Xi_L=H_L$ is the Poincar\'e -- Cartan form (\ref{303}). If
$u$ is an exact symmetry of $L$, i.e. $\bL_{J^1u}L=0$ we obtain a
weak conservation law
\mar{K4}\beq
0\ap - d_\la[\pi^\la_i(u^\m y^i_\m-u^i )-u^\la\cL] \label{K4}
\eeq
of a symmetry current
\mar{Q30}\beq
\cJ_u =[\pi^\la_i(u^\m y^i_\m-u^i )-u^\la\cL]\om_\la \label{Q30}
\eeq
along a vector field $u$ on the shell $\cE_L=0$ (\ref{b327}).

The weak conservation law (\ref{K4}) leads to a differential
conservation law
\be
\dr_\la(\cJ^\la_u\circ s)=0
\ee
on solutions $s$ of the Euler -- Lagrange equation (\ref{b327}).

It is readily observed that the symmetry current $\cJ_u$
(\ref{Q30}) is linear in a vector field $u$. Therefore, one can
consider a superposition of symmetry currents
\be
\cJ_u+\cJ_{u'}=\cJ_{u+u'}, \qquad \cJ_{cu}=c\cJ_u, \qquad
c\in\mathbb R,
\ee
and a superposition of weak conservation laws (\ref{K4})
associated to different symmetries $u$.

For instance, let $u=u^i\dr_i$ be  a vertical vector field on
$Y\to X$. If it is a symmetry of $L$, the weak conservation law
(\ref{K4}) takes a form
\mar{23f46}\beq
0\ap -d_\la(\pi^\la_i u^i). \label{23f46}
\eeq
It is called the Noether conservation law  of the Noether current
\mar{b374}\beq
\cJ_u =-\pi^\la_i u^i\om_\la \label{b374}
\eeq
along a Noether symmetry $u$.

Given the connection $\G$ (\ref{ps9}) on a fibre bundle $Y\to X$,
a vector field $\tau$ on $X$ gives rise to the projectable vector
field
\mar{b1.85}\beq
\G\tau =\tau^\la(\dr_\la +\G^i_\la\dr_i) \label{b1.85}
\eeq
on $Y$. The corresponding symmetry current (\ref{Q30}) along
$\G\tau$ reads
\mar{cc205'}\beq
\cJ_\G  =\tau^\m \cJ_\G{}^\la{}_\m=\tau^\m(\pi^\la_i(y_\m^i
-\G^i_\m)-\dl^\la_\m\cL)\om_\la. \label{cc205'}
\eeq
Its coefficients $\cJ_\G{}^\la{}_\m$ are components of the tensor
field
\mar{23f55}\beq
\cJ_\G=\cJ_\G{}^\la{}_\m dx^\m\ot\om_\la, \qquad \cJ_\G{}^\la{}_\m
=\pi^\la_i(y_\m^i-\G^i_\m)-\dl^\la_\m\cL,\label{23f55}
\eeq
called the energy-momentum tensor relative to a connection $\G$
\cite{book09,sard97,book13}. If $\G\tau$ (\ref{b1.85}) is a
symmetry of a Lagrangian $L$, we have the energy-momentum
conservation law
\mar{504}\beq
0 \ap  - d_\la [\pi^\la_i\tau^\m( y^i_\m-
\G^i_\m)-\dl^\la_\m\tau^\m\cL]. \label{504}
\eeq

Turn now to PS Hamiltonian formalism on a Legendre bundle $\Pi$.
Given the projectable vector field $u$ (\ref{ps40}) on $Y\to X$,
it gives rise to a vector field
\mar{Q4}\beq
\wt u=u^\m\dr_\m + u^i\dr_i +( - \dr_i u^j p^\la_j -\dr_\m u^\m
p^\la_i +\dr_\m u^\la p^\m_i)\dr^i_\la \label{Q4}
\eeq
on a Legendre bundle $\Pi\to Y$ and to a vector field
\mar{Q4'}\beq
J\wt u=\wt u +J^1u \label{Q4'}
\eeq
on $\Pi\op\times_Y J^1Y$. Then we have
\mar{b4180}\beq
\bL_{\wt u}H= \bL_{J\wt u}L_H=(-u^i\dr_i\cH -\dr_\m(u^\m\cH)
-u^\la_i\dr^i_\la\cH +p^\la_i\dr_\la u^i)\om. \label{b4180}
\eeq
It follows that a Hamiltonian form $H$ and a Lagrangian  $L_H$
have the same classical symmetries.

Let us  apply the first variational formula (\ref{23f42}) to the
Lie derivative $\bL_{J\wt u}L_H$ (\ref{Q3}). It reads
\be
&& -u^i\dr_i\cH -\dr_\m(u^\m\cH)
-u^\la_i\dr^i_\la\cH +p^\la_i\dr_\la u^i =
 -(u^i-y^i_\m u^\m)(p^\la_{\la i} +\dr_i\cH) + \\
&& \qquad ( - \dr_i u^j p^\la_j -\dr_\m u^\m p^\la_i
+\dr_\m u^\la p^\m_i- p^\la_{\m i}u^\m)(y^i_\la - \dr^i_\la\cH) -\\
&&\qquad  d_\la[p^\la_i(\dr^i_\m\cH u^\m-u^i) - u^\la (p^\m_i\dr^i_\m\cH
-\cH)].
\ee
On the shell (\ref{b4100a}) -- (\ref{b4100b}), this identity takes
a form
\mar{Q6}\beq
-u^i\dr_i\cH -\dr_\m(u^\m\cH) -u^\la_i\dr^i_\la\cH +p^\la_i\dr_\la
u^i \ap -    d_\la[p^\la_i(\dr^i_\m\cH u^\m-u^i) - u^\la
(p^\m_i\dr^i_\m\cH -\cH)]. \label{Q6}
\eeq
If $\bL_{J^1\wt u}L_H=0$, we obtain a weak Hamiltonian
conservation law
\mar{Q7'}\beq
0\ap  -
 d_\la[p^\la_i(u^\m\dr^i_\m\cH -u^i) - u^\la (p^\m_i \dr^i_\m\cH -\cH)]
\label{Q7'}
\eeq
of a Hamiltonian symmetry current
\mar{b4183}\beq
\wt\cJ_u = [p^\la_i(u^\m\dr^i_\m\cH -u^i) - u^\la (p^\m_i
\dr^i_\m\cH -\cH]\om_\la. \label{b4183}
\eeq

In particular, let $u=u^i\dr_i$ be a vertical vector field on
$Y\to X$. Then the Lie derivative $\bL_{\wt u}H$ (\ref{b4180})
takes a form
\be
\bL_{\wt u}H =(-u^i\dr_i\cH  +\dr_iu^jp^\la_j\dr^i_\la\cH
+p^\la_i\dr_\la u^i)\om.
\ee
The corresponding Noether Hamiltonian current (\ref{b4183}) reads
\mar{mos010}\beq
\wt\cJ_u =-u^i p^\la_i\om_\la \label{mos010}
\eeq
(cf. $\cJ_u$ (\ref{b374})). This is independent of a Hamiltonian
form $H$, and is defined only by a vertical vector field $u$. It
follows that Noether Hamiltonian currents $\wt\cJ_u$
(\ref{mos010}) in PS Hamiltonian theory constitute a real vector
space $\cJ(\Pi)$ isomorphic to that of vertical vector fields $u$
on a fibre bundle $Y$.

Moreover, due to the isomorphism (\ref{000}), Noether Hamiltonian
currents (\ref{mos010}) are represented by $TX$-valued densities
(\ref{ps11}):
\mar{ps44}\beq
\wt\cJ_u =-u^i p^\la_i\dr_\la\ot\om. \label{ps44}
\eeq
If $Y\to X$ is a vector bundle, the PS bracket $\{,\}_{PS}$
(\ref{xx3}) provides Noether Hamiltonian currents (\ref{ps44})
with the Lie bracket
\mar{ps45}\beq
[\wt\cJ_u,\wt\cJ_{u'}]=\{\wt\cJ_u,\wt\cJ_{u'}\}_{PS}=\wt\cJ_{[u,u']}
\label{ps45}
\eeq
which brings their space $\cJ(\Pi)$ into a real Lie algebra,
isomorphic to the Lie algebra of vertical vector fields on $Y$.
Similarly, the Poisson bracket $\{,\}_V$ (\ref{m72}) defines a Lie
bracket of Noether currents in mechanics \cite{book10}.

Let now $\tau =\tau^\la\dr_\la$ be a vector field on $X$ and
$\G\tau$ (\ref{b1.85}) its horizontal lift onto $Y$ by means of a
connection $\G$ on $Y\to X$. Given the splitting (\ref{4.7}) of a
Hamiltonian form $H$, the Lie derivative (\ref{b4180}) reads
\be
\bL_{\wt u}H =p^\la_j([\dr_\la +\G^i_\la\dr_i,u]^j -[\dr_\la
+\G^i_\la\dr_i,u]^\nu\G_\nu^j)\om -
 (\dr_\m u^\m\wt\cH_\G +u\rfloor
d\wt\cH_\G)\om,
\ee
where $[.,.]$ is the Lie bracket of vector fields. Then the weak
identity (\ref{Q6}) takes a form
\be
-(\dr_\m +\G^j_\m\dr_j - p^\la_i\dr_j\G^i_\m\dr^j_\la)\wt\cH_\G +
p^\la_iR^i_{\la\m} \ap - d_\la\wt\cJ_\G{}^\la{}_\m,
\ee
where
\be
R = \frac12 R_{\la\m}^i dx^\la\wedge dx^\m\otimes\dr_i,\qquad
R_{\la\m}^i = \dr_\la\G_\m^i - \dr_\m\G_\la^i + \G_\la^j\dr_j
\G_\m^i - \G_\m^j\dr_j \G_\la^i,
\ee
is the curvature of a connection $\G$. The corresponding symmetry
current (\ref{b4183}) reads
\mar{Q12}\beq
\wt\cJ_\G^\la
=\tau^\m\wt\cJ_\G{}^\la{}_\m=\tau^\m(p^\la_i\dr_\m^i\wt\cH_\G
-\dl^\la_\m(p^\nu_i\dr^i_\nu\wt\cH_\G -\wt\cH_\G)). \label{Q12}
\eeq
The relations (\ref{Q10'}) below show that, on a Lagrangian
constraint space $N_L$, the current (\ref{Q12}) can be treated as
a Hamiltonian energy-momentum current relative to a connection
$\G$.

On solutions $r$ of the covariant Hamilton equation (\ref{b4100a})
-- (\ref{b4100b}), the weak equality (\ref{Q7'}) leads to a
differential conservation law
\be
\dr_\la(\wt\cJ^\la_u (r)=0.
\ee

There is the following relation between differential conservation
laws in Lagrangian and PS Hamiltonian formalisms.

\begin{theorem}\label{hamlaw} \mar{hamlaw} Let a
Hamiltonian form $H$ be associated to an almost regular Lagrangian
$L$. Let $r$ be a solution of the covariant Hamilton equation
(\ref{b4100a}) -- (\ref{b4100b}) for $H$ which lives in a
Lagrangian constraint space $N_L$. Let $s=\pi_{\Pi Y}\circ r$ be
the corresponding solution  of the Euler -- Lagrange equation for
$L$ so that the relation (\ref{2.36}) holds. Then, for any
projectable vector field $u$ on a fibre bundle $Y\to X$, we have
\mar{Q10'}\beq
\wt\cJ_u (r)=\cJ_u( \pi_{\Pi Y}\circ r),\qquad \wt\cJ_u (\wh
L\circ J^1s) =\cJ_u(s), \label{Q10'}
\eeq
where $\cJ_u$ is the symmetry current (\ref{Q30}) on $J^1Y$  and
$\wt\cJ_u$ is the symmetry current (\ref{b4183}) on $\Pi$.
\end{theorem}

\begin{proof}
The proof follows from the relations (\ref{b481'}), (\ref{b481})
and (\ref{2.32}).
\end{proof}

By virtue of Theorems \ref{3.23} -- \ref{3.24}, it follows that:

$\bullet$ if $\cJ_u$ in Theorem \ref{hamlaw} is a conserved
symmetry current, then the symmetry current $\wt\cJ_u$
(\ref{Q10'}) is conserved on solutions of the Hamilton equation
which live in a Lagrangian constraint space,

$\bullet$ if $\wt\cJ_u$ in Theorem \ref{hamlaw} is a conserved
symmetry current, then the symmetry current $\cJ_u$ (\ref{Q10'})
is conserved on solutions $s$ of the Euler -- Lagrange equation
which obey the condition (\ref{2.36}).

However, Theorem \ref{hamlaw} fails to provide straightforward
relations between symmetries of Lagrangians and associated
Hamiltonian forms. In Section 12, we can obtain such  relations
between symmetries of constrained Lagrangians $L_N$ (\ref{cmp81})
and original quadratic Lagrangians $L$(\ref{N12}) (Theorem
\ref{y84}).

\section{Lagrangian and Hamiltonian Jacobi
fields}

The vertical extension of Lagrangian theory on a fibre bundle
$Y\to X$ onto the vertical tangent bundle $VY$ of $Y\to X$
describes the linear deviations of solutions of the Euler --
Lagrange equation which are Jacobi fields \cite{book09,jacobi}.
Accordingly, the vertical extension of PS Hamiltonian formalism on
the Legendre bundle $\Pi$ (\ref{00}) onto the vertical Legendre
bundle $\Pi_{VY}$ (\ref{v3}) describes Jacobi fields of solutions
of the covariant Hamilton equations \cite{book09,book00}.

Let $VY$ be the vertical tangent bundle of $Y\to X$ endowed with
holonomic coordinates $(x^\la, y^i,\dot y^i)$. The configuration
space of first order Lagrangian theory on $VY\to X$ is the jet
manifold $J^1VY$. There is the canonical isomorphism
\be
J^1VY\op=_{J^1Y} VJ^1Y, \qquad \dot y^i_\la=(\dot y^i)_\la,
\ee
where, in comparison with $V_YJ^1Y$ in the expression (\ref{ps2}),
$VJ^1Y$ is the vertical tangent bundle of $J^1Y\to X$ which is
provided with holonomic coordinates $(x^\la, y^i, y^i_\la, \dot
y^i, \dot y^i_\la)$. Due to this isomorphism, first order
Lagrangian formalism on $VY$ can be developed as the vertical
extension of Lagrangian theory on $Y$.

\begin{lemma}
Similar to the canonical isomorphism between fibre bundles $TT^*Z$
and $T^*TZ$ \cite{2kij}, the isomorphism
\mar{v10}\beq
VV^*Y \op=_{VY} V^*VY, \qquad p_i\llra\dot v_i, \quad \dot
p_i\llra\dot y_i,\label{v10}
\eeq
can be established by inspection of the transformation laws of
holonomic coordinates $(x^\la, y^i, p_i=\ol y_i)$ on $V^*Y$ and
$(x^\la, y^i, v^i=\dot y^i)$ on $VY$.
\end{lemma}

It follows that any exterior form $\f$ on a fibre bundle $Y$ gives
rise to an exterior form
\mar{ws531x}\beq
\f_V=\dr_V(\f)=\dot y^i\dr_i(\f) \label{ws531x}
\eeq
on $VY$ so that $d\f_V=(d\f)_V$. For instance,
\be
\dr_Vf=\dot y^i\dr_if, \qquad \dr_V(dy^i)=d\dot y^i, \qquad f\in
C^\infty(Y).
\ee
The form $\f_V$ (\ref{ws531x}) is called the vertical extension of
$\f$ on $Y$.

Let $L$ be the Lagrangian (\ref{23f2}) on $J^1Y$. Its vertical
extension $L_V$ (\ref{ws531x}) onto $VJ^1Y$ (but not $VL$
(\ref{ps2}) onto $V_YJ^1Y$) reads
\mar{v0}\beq
L_V=\dr_VL=(\dot y^i\dr_i + \dot y^i_\la\dr_i^\la)\cL\om.
\label{v0}
\eeq

The corresponding Euler -- Lagrange equation (\ref{b327}) takes a
form
\mar{v1,2}\ben
&& \dot\dl_i\cL_V=\dl_i\cL=0, \label{v1}\\
&& \dl_i\cL_V=\dr_V\dl_i\cL=0, \label{v2}\\
&& \dr_V= \dot y^i\dr_i + \dot y^i_\la\dr_i^\la+ \dot
y^i_{\m\la}\dr_i^{\m\la}. \nonumber
\een
The equation (\ref{v1}) is exactly the Euler -- Lagrange equation
(\ref{b327}) for an original Lagrangian $L$. In order to clarify
the meaning of the equation (\ref{v2}), let us suppose that $Y\to
X$ is a vector bundle. Given a solution $s$ of the Euler --
Lagrange equation (\ref{v1}), let $\dl s$ be a Jacobi field, i.e.,
$s+\ve \dl s$ also is a solution of the Euler -- Lagrange equation
(\ref{v1}) modulo the terms of order $>1$ in a small parameter
$\ve$. Then it is readily observed that a Jacobi field $\dl s$
satisfies the Euler -- Lagrange equation (\ref{v2}), which
therefore is called the variation equation of the equation
(\ref{v1}) \cite{ditt,book00, jacobi}.

The Lagrangian $L_V$ (\ref{v0}) yields a Legendre map
\mar{v3}\beq
\wh L_V: VJ^1Y\ar_{VY} \Pi_{VY}=V^*VY\op\w_{VY}(\op\w^{n-1} T^*X),
\label{v3}
\eeq
where $\Pi_{VY}$ is called the vertical Legendre bundle.

\begin{lemma}
Due to the isomorphism (\ref{v10}) there exists a bundle
isomorphism
\mar{cmp42}\beq
\Pi_{VY}\op=_{VY} V\Pi, \qquad
 p^\la_i\llra\dot p^\la_i,
\qquad q^\la_i\llra p^\la_i, \label{cmp42}
\eeq
written with respect to the holonomic coordinates $(x^\la, y^i,
\dot y^i, p^\la_i, q^\la_i)$ on $\Pi_{VY}$ and  $(x^\la, y^i,
p^\la_i, \dot y^i,\dot p^\la_i)$ on $V\Pi$.
\end{lemma}

In view of the isomorphism (\ref{cmp42}), the Legendre map
(\ref{v3}) takes a form
\mar{v4}\ben
&& \wh L_V=V\wh L: VJ^1Y\ar_{VY} \Pi_{VY}=V\Pi, \label{v4}\\
&& p^\la_i=\dot\dr^\la_i\cL_V=\pi^\la_i, \qquad \dot
p^\la_i=\dr^i_\la\cL=\dr_V\pi^\la_i. \nonumber
\een
It is called the vertical Legendre map.

Let $Z_{VY}$ be the homogeneous Legendre bundle (\ref{N41}) over
$VY$ endowed with the corresponding coordinates $(x^\la,y^i,\dot
y^i,p_i^\la,q_i^\la,p)$. There is a fibre bundle
\mar{cmp40}\beq
\zeta: VZ_Y\to Z_{VY}, \qquad (x^\la,y^i,\dot
y^i,p_i^\la,q_i^\la,p) \circ\zeta= (x^\la,y^i,\dot y^i,\dot
p_i^\la, p_i^\la,\dot p). \label{cmp40}
\eeq
Then the vertical tangent morphism $V\pi_{Z\Pi}$ to $\pi_{Z\Pi}$
(\ref{b418'}) factorizes through the composition of fibre bundles
\mar{cmp34}\beq
V\pi_{Z\Pi}: VZ_Y\to Z_{VY}\to \Pi_{VY}=V\Pi. \label{cmp34}
\eeq

Owing to this fact, one can develop PS Hamiltonian formalism on a
momentum phase space $\Pi_{VY}$ as the vertical extension of PS
Hamiltonian theory on $\Pi$. The corresponding canonical conjugate
pairs are $(y^i,\dot p^\la_i)$ and $(\dot y^i,p_i^\la)$. In
particular, due to the isomorphism (\ref{cmp42}), $V\Pi$ is
endowed with the canonical PS form (\ref{406}) which reads
\be
\bom_{VY}=[d\dot p^\la_i\w dy^i +dp^\la_i\w d\dot
y^i]\w\om\ot\dr_\la.
\ee

Let $Z_{VY}$ be the homogeneous Legendre bundle (\ref{N41}) over
$VY$ with the corresponding coordinates $(x^\la,y^i,\dot
y^i,p_i^\la,q_i^\la,p)$. It can be endowed with the
multisymplectic Liouville form $\Xi_{VY}$ (\ref{N43}). Sections of
the affine bundle
\mar{cmp41}\beq
 Z_{VY}\to V\Pi, \label{cmp41}
\eeq
by definition, provide Hamiltonian forms on $V\Pi$.

Let us consider the following particular case of these forms which
are related to those on a Legendre bundle $\Pi$. Due to the fibre
bundle (\ref{cmp40}):
\be
\zeta: VZ_Y\to Z_{VY},
\ee
the vertical tangent bundle $VZ_Y$ of $Z_Y\to X$ is provided with
an exterior form
\be
\Xi_V=\zeta^*\Xi_{VY}= \dot p\om + (\dot p^\la_i dy^i +
p^\la_id\dot y^i) \w\om_\la,
\ee
which is exactly the vertical extension (\ref{ws531x}) of the
canonical multisymplectic Liouville form $\Xi$ on $Z_Y$. Given the
affine bundle $Z_Y\to\Pi$ (\ref{b418'}), we have the fibre bundle
$VZ_Y\to V\Pi$ (\ref{cmp34}) where $V\pi_{Z\Pi}$ is the vertical
tangent map to $\pi_{Z\Pi}$. Let $h$ be a section of an affine
bundle $Z_Y\to \Pi$ and $H=h^*\Xi$ the corresponding Hamiltonian
form (\ref{b418}) on $\Pi$. Then a section $Vh$ of the fibre
bundle (\ref{cmp34}) and the corresponding section $\zeta\circ Vh$
of the affine bundle (\ref{cmp41}) defines the Hamiltonian form
\mar{m17}\ben
&& H_V=(Vh)^*\Xi_V =(\dot p^\la_idy^i + p^\la_i d\dot
y^i)\w\om_\la -\cH_V\om,
\label{m17}\\
&& \cH_V =\dr_V\cH, \qquad \dr_V=\dot y^i\dr_i +\dot
p^\la_i\dr^i_\la, \nonumber
\een
on $V\Pi$. It is called the  vertical extension of $H$ (or,
simply, the vertical Hamiltonian form). In particular, given the
splitting (\ref{4.7}) of $H$ with respect to a connection $\G$ on
$Y\to X$, we have the corresponding splitting
\be
\cH_V=\dot p^\la_i\G^i_\la +\dot y^j p^\la_i\dr_j\G^i_\la
+\dr_V\wt\cH_\G
\ee
of $H_V$ with respect to the canonical vertical prolongation
\mar{43}\beq
V\G : VY\to J^1VY, \qquad  V\G = dx^\la\otimes(\dr_\la
+\G^i_\la\dr_i+\dr_j\G^i_\la\dot y^j \dot\dr_i), \label{43}
\eeq
of $\G$ onto $VY\to X$.

\begin{theorem}
Let $\g$ (\ref{cmp33}) be a Hamiltonian connection on $\Pi$
associated to a Hamiltonian form $H$. Then its vertical
prolongation $V\g$ (\ref{43}) on $V\Pi\to X$ is a Hamiltonian
connection associated to the vertical Hamiltonian form $H_V$
(\ref{m17}).
\end{theorem}

\begin{proof}
The proof follows from a direct computation. We have
\be
V\g=\g + dx^\m\ot [\dr_V\g^i_\m\dot\dr_i +\dr_V\g^\la_{\m
i}\dot\dr_\la^i].
\ee
Components of this connection obey the equation
\mar{cmp51}\beq
\dot \g^i_\m=\dr^i_\m\cH_V=\dr_V\dr^i_\m\cH,\qquad
 \dot \g^\la_{\la i}=-\dr_i\cH_V=-\dr_V\dr_i\cH \label{cmp51}
\eeq
and the equation (\ref{cmp3}).
\end{proof}

In order to clarify the meaning of the equation (\ref{cmp51}), let
us suppose that $Y\to X$ is a vector bundle. Given a solution $r$
of the Hamilton equation (\ref{b4100a}) -- (\ref{b4100b}) for $H$,
let $\ol r$ be a Jacobi field, i.e., $r+\ve \ol r$ also is a
solution of the same Hamilton equation modulo terms of order $>1$
in $\ve$. Then it is readily observed that a Jacobi field $\ol r$
satisfies the equation (\ref{cmp51}). At the same time, the
Lagrangian $L_{H_V}$ (\ref{Q3}) on $J^1V\Pi$, defined by the
Hamiltonian form $H_V$ (\ref{m17}),  takes a form
\mar{cmp105}\beq
\cL_{H_V}=h_0(H_V)=\dot p^\la_i(y^i_\la- \dr^i_\la\cH) -\dot
y^i(p^\la_{\la i} + \dr_i\cH) +d_\la(p^\la_i\dot y^i),
\label{cmp105}
\eeq
where $\dot p^\la_i$, $\dot y^i$ play a role of the Lagrange
multipliers.

In conclusion, let us study the relationship between the vertical
extensions of Lagrangian and PS Hamiltonian formalisms. The
Hamiltonian form $H_V$ (\ref{m17}) on $V\Pi$ yields the vertical
Hamiltonian map
\be
&& \wh H_V=V\wh H: V\Pi\op\to_{VY} VJ^1Y, \\
&& y^i_\la=\dot\dr^i_\la\cH_V =\dr^i_\la\cH, \qquad
\dot y^i_\la= \dr_V\dr^i_\la\cH.
\ee

\begin{theorem} \label{p02} \mar{p02} Let a Hamiltonian form $H$ on $\Pi$
be associated to a Lagrangian $L$ on $J^1Y$. Then the vertical
Hamiltonian form $H_V$ (\ref{m17}) is weakly associated to the
Lagrangian $L_V$ (\ref{v0}).
\end{theorem}

\begin{proof}
If the morphisms $\wh H$ and $\wh L$ obey the relation
(\ref{2.30a}), then the corresponding vertical tangent morphisms
satisfy the relation
\be
V\wh L\circ V\wh H\circ Vi_N=Vi_N.
\ee
The condition (\ref{2.30b}) for $H_V$ reduces to the equality
(\ref{2.31}) which is fulfilled if $H$ is associated to $L$.
\end{proof}

\section{Quadratic Lagrangian and Hamiltonian systems}

Field theories with almost regular quadratic Lagrangians admit
comprehensive PS Hamiltonian formulation \cite{book,jpa99,book09}.

Given a fibre bundle $Y\to X$, let us consider the  quadratic
Lagrangian $L$ (\ref{23f2}):
\mar{N12}\beq
\cL=\frac12 a^{\la\m}_{ij}(x^\nu,y^k) y^i_\la y^j_\m +
b^\la_i(x^\nu,y^k) y^i_\la + c(x^\nu,y^k), \label{N12}
\eeq
where $a$, $b$ and $c$ are local functions on $Y$. This property
is coordinate-independent due to the affine transformation law
(\ref{50}) of coordinates $y^i_\la$. The associated Legendre map
$\wh L$ (\ref{b330}) is given by the coordinate expression
\mar{N13}\beq
p^\la_i\circ\wh L= a^{\la\m}_{ij} y^j_\m +b^\la_i, \label{N13}
\eeq
and is an affine morphism over $Y$. It yields the corresponding
linear morphism
\mar{N13'}\beq
\wh a: T^*X\op\otimes_YVY\op\to_Y \ol N_L\subset \Pi,\qquad
p^\la_i\circ\wh a=a^{\la\m}_{ij}\ol y^j_\m, \label{N13'}
\eeq
where $\ol y^j_\mu$ are fibred coordinates on the vector bundle
(\ref{cc9}).

Let the Lagrangian $L$ (\ref{N12}) be almost regular, i.e., the
morphism $\wh a$ (\ref{N13'}) is of constant rank. Then the
Lagrangian constraint space $N_L$ (\ref{N13}) is an affine
subbundle of the Legendre bundle $\Pi\to Y$, modelled over the
vector subbundle $\ol N_L$ (\ref{N13'}) of $\Pi\to Y$. Hence,
$N_L\to Y$ has a global section $s$. For the sake of simplicity,
let us assume that $s=\wh 0$ is the canonical zero section of
$\Pi\to Y$. Then $\ol N_L=N_L$. Accordingly, the kernel of the
Legendre map (\ref{N13})  is an affine subbundle of the affine jet
bundle $J^1Y\to Y$, modelled over the kernel of the linear
morphism $\wh a$ (\ref{N13'}). Then there exists a connection
\mar{N16}\beq
\G: Y\to \Ker\wh L\subset J^1Y, \qquad a^{\la\m}_{ij}\G^j_\m +
b^\la_i =0, \label{N16}
\eeq
on $Y\to X$. Connections (\ref{N16}) constitute an affine space
modelled over a vector space of soldering forms
\be
\f=\f^i_\la dx^\la\ot\dr_i
\ee
on $Y\to X$, satisfying the conditions
\mar{cmp21}\beq
a^{\la\m}_{ij}\f^j_\m =0 \label{cmp21}
\eeq
and, as a consequence, the conditions $\f^i_\la b^\la_i=0$. If the
Lagrangian (\ref{N12}) is regular, the connection (\ref{N16}) is
unique.

\begin{remark}
If $s\neq\wh 0$, one can consider connections $\G$ taking their
values into $\Ker_s\wh L$.
\end{remark}

A matrix $a$ in the Lagrangian $L$ (\ref{N12}) can be seen as a
global section of constant rank of a tensor bundle
\be
\op\w^n T^*X\op\ot_Y[\op\vee^2(TX\op\ot_Y V^*Y)]\to Y.
\ee
 Then it satisfies the following corollary of the well-known theorem on a splitting of
 a short exact sequence of vector bundles \cite{book09}.

\begin{corollary} \label{mm45} \mar{mm45}
Given a $k$-dimensional vector bundle $E\to Z$, let $a$ be a fibre
metric  of rank $r$ in $E$. There is a splitting
\mar{mm50}\beq
E= \Ker a\op\oplus_Z  E' \label{mm50}
\eeq
where $E'=E/\Ker a$ is the quotient bundle, and $a$ is a
non-degenerate fibre metric in $E'$.
\end{corollary}

\begin{theorem}\label{04.2} \mar{04.2} There exists a linear bundle
map
\mar{N17}\beq
\si: \Pi\op\to_Y T^*X\op\otimes_YVY, \qquad \ol y^i_\la\circ\si
=\si^{ij}_{\la\m}p^\m_j, \label{N17}
\eeq
such that $\wh a\circ\si\circ i_N= i_N$.
\end{theorem}

\begin{proof}
The map (\ref{N17}) is a solution of algebraic equations
\mar{N45}\beq
a^{\la\mu}_{ij}\si^{jk}_{\mu\al}a^{\al\nu}_{kb}=a^{\la\nu}_{ib}.
\label{N45}
\eeq
By virtue of Corollary \ref{mm45}, there exists the bundle
splitting
\mar{mm46}\beq
TX^*\op\ot_Y VY=\Ker a\op\oplus_Y E' \label{mm46}
\eeq
and an atlas of this bundle such that transition functions of
$\Ker a$ and $E'$ are mutually independent. Since $a$ is a
non-degenerate section of
\be
\op\w^n T^*X\op\ot_Y(\op\vee^2E'^*)\to Y,
\ee
there exist fibre coordinates $(\ol y^A)$ on $E'$ such that $a$ is
brought into a diagonal matrix with non-vanishing components
$a_{AA}$. Due to the splitting (\ref{mm46}), we have the
corresponding bundle splitting
\be
TX\op\ot_Y V^*Y=(\Ker a)^*\op\oplus_Y E'^*.
\ee
Then a desired map $\si$ is represented by a direct sum
$\si_1\oplus\si_0$ of an arbitrary section $\si_1$ of a fibre
bundle
\be
\op\w^n TX\op\ot_Y(\op\vee^2\Ker a)\to Y
\ee
and the section $\si_0$ of a fibre bundle
\be
\op\w^n TX\op\ot_Y(\op\vee^2E')\to Y
\ee
which has non-vanishing components $\si^{AA}=(a_{AA})^{-1}$ with
respect to the fibre coordinates $(\ol y^A)$ on $E'$. We have
relations
\mar{N21}\beq
\si_0=\si_0\circ a\circ\si_0, \qquad a\circ\si_1=0, \qquad
\si_1\circ a=0. \label{N21}
\eeq
\end{proof}

\begin{remark}
Using the relations (\ref{N21}), one can write the above
assumption, that the Lagrangian constraint space $N_L\to Y$ admits
a global zero section, in the form
\mar{NN21}\beq
b^\m_i=a^{\m\la}_{ij}\si^{jk}_{\la\nu} b^\nu_k. \label{NN21}
\eeq
\end{remark}

With the relations (\ref{N16}), (\ref{N45}) and (\ref{N21}), we
obtain a splitting
\mar{N18,b4122}\ben
&& J^1Y=\cS(J^1Y)\op\oplus_Y \cF(J^1Y)=\Ker\wh L\op\oplus_Y{\rm
Im}(\si\circ
\wh L), \label{N18} \\
&& y^i_\la=\cS^i_\la+\cF^i_\la= [y^i_\la
-\si^{ik}_{\la\al} (a^{\al\m}_{kj}y^j_\m + b^\al_k)]+
[\si^{ik}_{\la\al} (a^{\al\m}_{kj}y^j_\m + b^\al_k)],
\label{b4122}
\een
where, in fact,  $\si=\si_0$ owing to the relations (\ref{N21})
and (\ref{NN21}). Then with respect to the coordinates $\cS^i_\la$
and $\cF^i_\la$ (\ref{b4122}), the Lagrangian (\ref{N12}) reads
\mar{cmp31}\beq
\cL=\frac12 a^{\la\m}_{ij}\cF^i_\la\cF^j_\m +c', \label{cmp31}
\eeq
where
\mar{mos018}\beq
\cF^i_\la= \si_0{}^{ik}_{\la\al} a^{\al\m}_{kj}(y^j_\m -\G^j_\m)
\label{mos018}
\eeq
for some $(\Ker\wh L)$-valued connection $\G$ (\ref{N16}) on $Y\to
X$. Thus, the Lagrangian (\ref{N12}), written in the form
(\ref{cmp31}), factorizes through the covariant differential
relative to any such connection.

Turn now to PS Hamiltonian formalism. Let $L$ (\ref{N12}) be an
almost regular quadratic Lagrangian brought into the form
(\ref{cmp31}), $\si=\si_0+\si_1$ the linear map (\ref{N17}) and
$\G$ the connection (\ref{N16}). Similarly to the splitting
(\ref{N18}) of a configuration space $J^1Y$, we have the following
decomposition of a momentum phase space:
\mar{N20,'}\ben
&& \Pi=\cR(\Pi)\op\oplus_Y\cP(\Pi)=\Ker\si_0 \op\oplus_Y N_L, \label{N20}
\\ && p^\la_i = \cR^\la_i+\cP^\la_i= [p^\la_i -
a^{\la\m}_{ij}\si^{jk}_{\m\al}p^\al_k] +
[a^{\la\m}_{ij}\si^{jk}_{\m\al}p^\al_k]. \label{N20'}
\een
The relations (\ref{N21}) lead to the equalities
\mar{m25}\beq
\si_0{}^{jk}_{\m\al}\cR^\al_k=0, \qquad
\si_1{}^{jk}_{\m\al}\cP^\al_k=0, \qquad \cR^\la_i\cF^i_\la=0.
\label{m25}
\eeq
Relative to the coordinates (\ref{N20'}), the Lagrangian
constraint space $N_L$ (\ref{N13}) is given by the equations
\mar{zzz}\beq
\cR^\la_i= p^\la_i - a^{\la\m}_{ij}\si^{jk}_{\m\al}p^\al_k=0.
\label{zzz}
\eeq

Let the splitting (\ref{mm46}) be provided with adapted fibre
coordinates $(\ol y^a,\ol y^A)$ such that the matrix function $a$
(\ref{N13'}) is brought into a diagonal matrix with non-vanishing
components $a_{AA}$. Then the Legendre bundle $\Pi$ (\ref{N20}) is
endowed with the dual (non-holonomic) fibre coordinates
$(p_a,p_A)$ where $p_A$ are coordinates on a Lagrangian constraint
space $N_L$, given by the equalities $p_a=0$. Relative to these
coordinates, $\si_0$ becomes the diagonal matrix
\mar{m39}\beq
\si_0^{AA}=(a_{AA})^{-1}, \qquad \si_0^{aa}=0, \label{m39}
\eeq
while $\si_1^{Aa}=\si_1^{AB}=0$. Let us write
\mar{m41}\beq
p_a=M_a{}^i_\la p^\la_i, \qquad p_A=M_A{}^i_\la p^\la_i,
\label{m41}
\eeq
where $M$ are the matrix functions on $Y$ which obeys the
relations
\mar{y30}\ben
&& M_a{}^i_\la a^{\la\m}_{ij}=0, \qquad
(M^{-1})^a{}_i^\la\si_0{}^{ij}_{\la\m}=0, \label{y30}\\
&& M_A{}^i_\la(a\circ\si_0)^{\la j}_{i\m}= M_A{}^j_\m, \qquad
(M^{-1})^A{}_j^\m M_A{}^i_\la=
a_{jk}^{\m\nu}\si_0{}^{ki}_{\nu\la}. \nonumber
\een

Let us consider the affine Hamiltonian map
\mar{N19}\beq
\Phi=\wh\G+\si:\Pi \op\to J^1Y, \qquad \Phi^i_\la = \G^i_\la +
\si^{ij}_{\la\m}p^\m_j,\label{N19}
\eeq
and the Hamiltonian form
\mar{N22}\ben
&& H(\G,\si_1)=H_\Phi +\Phi^*L= p^\la_idy^i\w\om_\la - \label{N22} \\
&& \qquad [\G^i_\la p^\la_i +\frac12 \si_0{}^{ij}_{\la\m}p^\la_ip^\m_j
+\si_1{}^{ij}_{\la\m}p^\la_ip^\m_j -c']\om = \nonumber\\
&& \qquad (\cR^\la_i+\cP^\la_i)dy^i\w\om_\la - [(\cR^\la_i+\cP^\la_i)\G^i_\la +\frac12
\si_0{}^{ij}_{\la\m}\cP^\la_i\cP^\m_j
+\si_1{}^{ij}_{\la\m}\cR^\la_i\cR^\m_j -c']\om.\nonumber
\een

\begin{theorem} \label{cmp30} \mar{cmp30} The Hamiltonian
forms $H(\G,\si_1)$ (\ref{N22}) parameterized by connections $\G$
(\ref{N16}) are weakly associated to the Lagrangian (\ref{N12}),
and they constitute a complete set.
\end{theorem}

\begin{proof}
By the very definitions of $\G$ and $\si$, the Hamiltonian map
(\ref{N19}) satisfies the condition (\ref{2.30a}). Then
$H(\G,\si_1)$ is weakly associated to $L$ (\ref{N12}) in
accordance with Theorem \ref{jp}. Let us write the corresponding
Hamilton equation (\ref{b4100a}) for a section $r$ of a Legendre
bundle $\Pi\to X$. It reads
\mar{N29}\beq
J^1s= (\wh\G+\si)\circ r, \qquad s=\pi_{\Pi Y}\circ r. \label{N29}
\eeq
Due to the surjections $\cS$ and $\cF$ (\ref{N18}), the Hamilton
equation (\ref{N29}) is brought into the two parts
\mar{N23,8}\ben
&&\cS\circ J^1s=\G\circ s, \qquad \dr_\la r^i- \si_0{}^{ik}_{\la\al} (a^{\al\m}_{kj}\dr_\mu r^j +
b^\al_k)=\G^i_\la\circ s, \label{N23}\\
&&\cF \circ J^1s=\si\circ r, \qquad \si_0{}^{ik}_{\la\al} (a^{\al\m}_{kj}\dr_\mu r^j + b^\al_k)=
\si^{ik}_{\la\al}r^\al_k. \label{N28}
\een
Let $s$ be an arbitrary section of $Y\to X$, e.g., a solution of
the Euler -- Lagrange equation. There exists the connection $\G$
(\ref{N16}) such that the relation (\ref{N23}) holds, namely,
$\G={\cal S}\circ\G'$ where $\G'$ is a connection on $Y\to X$
which has $s$ as an integral section. It is easily seen that, in
this case, the Hamiltonian map (\ref{N19}) satisfies the relation
(\ref{2.36}) for $s$. Hence, the Hamiltonian forms (\ref{N22})
constitute a complete set.
\end{proof}

It is readily observed that, if $\si_1=0$, then $\Phi=\wh H(\G)$,
and the Hamiltonian forms $H(\G,\si_1=0)$ (\ref{N22}) are
associated to the Lagrangian (\ref{N12}). For different $\si_1$,
we have different complete sets of Hamiltonian forms (\ref{N22}).
Hamiltonian forms $H(\G,\si_1)$ and $H(\G',\si_1)$ (\ref{N22}) of
such a complete set differ from each other in the term
$\f^i_\la\cR^\la_i$, where $\f$ are the soldering forms
(\ref{cmp21}). This term vanishes on the Lagrangian constraint
space (\ref{zzz}). Accordingly, the covariant Hamilton equations
for different Hamiltonian forms $H(\G,\si_1)$ and $H(\G',\si_1)$
(\ref{N22}) differ from each other in the equations (\ref{N23}).

Since the Lagrangian constraint space $N_L$ (\ref{zzz}) is an
imbedded subbundle of $\Pi \to Y$, all Hamiltonian forms
$H(\G,\si_1)$ (\ref{N22}) define a unique constrained Hamiltonian
form $H_N$ (\ref{b4300}) on $N_L$ which reads
\mar{94f1}\beq
H_N=i_N^*H(\G,\si_1)= \cP^\la_i dy^i\w\om_\la - [\cP^\la_i\G^i_\la
+\frac12 \si_0{}^{ij}_{\la\m}\cP^\la_i\cP^\m_j
 -c']\om.\label{94f1}
\eeq
In view of the relations (\ref{m25}), the corresponding
constrained Lagrangian $L_N$ (\ref{cmp81}) on $J^1N_L$ takes a
form
\mar{bv2}\beq
L_N=h_0(H_N)=(\cP^\la_i\cF_\la^i
-\frac12\si_0{}^{ij}_{\la\m}\cP^\la_i\cP^\m_j + c')\om.
\label{bv2}
\eeq
It is the pull-back onto $J^1N_L$ of a Lagrangian
\mar{m16}\beq
L_{H(\G,\si_1)}=\cR^\la_i(\cS^i_\la-\G^i_\la) +\cP^\la_i\cF_\la^i
-\frac12\si_0{}^{ij}_{\la\m}\cP^\la_i\cP^\m_j -
\frac12\si_1{}^{ij}_{\la\m} \cR^\la_i \cR^\m_j+ c' \label{m16}
\eeq
on $J^1\Pi$ for any Hamiltonian form $H(\G,\si_1)$ (\ref{N22}).

In fact, the Lagrangian $L_N$ (\ref{bv2}) is defined on the
product $N_L\times_Y J^1Y$ (see Remark \ref{91r10}). Since the
momentum phase space $\Pi$ (\ref{N20}) is a trivial bundle ${\rm
pr}_2:\Pi\to N_L$ over the Lagrangian constraint space $N_L$, one
can consider the pull-back
\mar{m32}\beq
L_\Pi=(\cP^\la_i\cF_\la^i -\frac12\si_0{}^{ij}_{\la\m}\cP^\la_i
\cP^\m_j + c')\om \label{m32}
\eeq
of the constrained Lagrangian $L_N$ (\ref{bv2}) onto
$\Pi\op\times_Y J^1Y$.

In a case of quadratic Lagrangians, we can improve Theorem
\ref{cmp22} as follows.

\begin{theorem} \label{cmp23} For every Hamiltonian
form $H$ (\ref{N22}), the Hamilton equations (\ref{b4100b}) and
(\ref{N28}) restricted to a Lagrangian constraint space $N_L$ are
equivalent to the constrained Hamilton equation (\ref{N44}).
\end{theorem}

\begin{proof} Due to the splitting (\ref{N20}), we have the corresponding
splitting of the vertical tangent bundle $V_Y\Pi$ of a Legendre
bundle $\Pi\to Y$. In particular, any vertical vector field $u$ on
$\Pi\to X$ admits the decomposition
\be
u= [u-u_{TN}] + u_{TN},  \qquad  u_{TN}=u^i\dr_i
+a^{\la\m}_{ij}\si^{jk}_{\m\al}u^\al_k\dr_\la^i,
\ee
such that $u_N=u_{TN}\mid_{N_L}$ is a vertical vector field on a
Lagrangian constraint space $N_L\to X$. Let us consider the
equations
\beq
r^*(u_{TN}\rfloor dH)=0 \label{cmp15}
\eeq
where $r$ is a section of $\Pi\to X$ and $u$ is an arbitrary
vertical vector field on $\Pi\to X$. They are equivalent to the
pair of equations
\ben
&& r^*(a^{\la\m}_{ij}\si^{jk}_{\m\al}\dr_\la^i\rfloor dH)=0,
\label{b4125a} \\
&& r^*(\dr_i\rfloor dH)=0. \label{b4125b}
\een
The equation (\ref{b4125b}) obviously is the Hamilton equation
(\ref{b4100b}) for $H$. Bearing in mind the relations (\ref{N16})
and (\ref{N21}), one can easily show that the equation
(\ref{b4125a}) coincides with the Hamilton equation (\ref{N28}).
The proof is completed by observing that, restricted to a
Lagrangian constraint space $N_L$, the equation (\ref{cmp15}) is
exactly the constrained Hamilton equation (\ref{N44'}).
\end{proof}

Theorem \ref{cmp23} shows that, restricted to a Lagrangian
constraint space, the Hamilton equation for different Hamiltonian
forms (\ref{N22}) associated to the same quadratic Lagrangian
(\ref{N12}) differ from each other in the equations (\ref{N23}).
These equations are independent of momenta and play a role of the
gauge-type conditions as follows.

By virtue of Theorem \ref{3.01}, the constrained Hamilton equation
is quasi-equivalent to the Cartan equation. A section $\ol s$ of
$J^1Y\to X$ is a solution of the Cartan equation for an almost
regular quadratic  Lagrangian (\ref{N12}) iff $r=\wh L\circ \ol s$
is a solution of the Hamilton equations (\ref{b4100b}) and
(\ref{N28}). In particular, let $\ol s$ be such a solution of the
Cartan equation and $\ol s_0$ a section of a fibre bundle
$T^*X\op\ot_Y VY\to X$ which takes its values into $\Ker \ol L$
(see (\ref{N13'})) and projects onto a section $s=\pi^1_0\circ \ol
s$ of $Y\to X$. Then the affine sum $\ol s +\ol s_0$ over
$s(X)\subset Y$  is also a solution of the Cartan equation. Thus,
we come to the notion of a gauge-type freedom of the Cartan
equation for an almost regular quadratic Lagrangian $L$. One can
speak of the gauge classes of solutions of the Cartan equation
whose elements differ from each other in the above-mentioned
sections $\ol s_0$. Let $z$ be such a gauge class whose elements
project onto a section $s$ of $Y\to X$. For different connections
$\G$ (\ref{N16}), we consider a condition
\beq
\cS\circ\ol s=\G\circ s, \qquad \ol s\in z. \label{cmp25}
\eeq

\begin{lemma} \label{cmp26}
(i) If two elements $\ol s$ and $\ol s'$ of the same  gauge class
$z$ obey the same condition (\ref{cmp25}), then $\ol s=\ol s'$.
(ii) For any solution $\ol s$ of the Cartan equation, there exists
a connection (\ref{N16}) which fulfills the condition
(\ref{cmp25}).
\end{lemma}

\begin{proof}
(i) Let us consider the affine difference  $\ol s-\ol s'$ over
$s(X)\subset Y$. We have $\cS(\ol s-\ol s')=0$ iff $\ol s=\ol s'$.
(ii) In the proof of Theorem \ref{cmp30}, we have shown that,
given $s=\pi^0_1\circ\ol s$, there exists the connection $\G$
(\ref{N16}) which fulfills the relation (\ref{N23}). Let us
consider the affine difference $\cS(\ol s- J^1s)$ over
$s(X)\subset Y$. This is a local section of the vector bundle
$\Ker \ol L\to Y$ over $s(X)$. Let $\f$ be its prolongation onto
$Y$. It is easy to see that $\G+\f$ is a desired connection.
\end{proof}

Due to the properties in Lemma \ref{cmp26}, one can treat
(\ref{cmp25}) as a gauge-type condition on solutions of the Cartan
equation. The Hamilton equation (\ref{N23}) exemplifies this
gauge-type condition when $\ol s= J^1s$ is a solution of the Euler
-- Lagrange equation. At the same time, the above-mentioned
freedom characterizes solutions of the Cartan equation, but not of
the Euler -- Lagrange one. First of all, this freedom reflects the
degeneracy of the Cartan equation (\ref{b336a}). Therefore, e.g.,
in Hamiltonian gauge theory (Section 13), the above mentioned
freedom is not related directly to the familiar gauge invariance.
Nevertheless, the Hamilton equation (\ref{N23}) are not gauge
invariant, and thus can play a role of gauge conditions in gauge
theory.

Now let us study symmetries of the Lagrangians $L_N$ (\ref{bv2})
and $L_\Pi$ (\ref{m32}) \cite{bs04,book09}. We aim to show that,
under certain conditions, they inherit Noether (i.e. vertical
classical) symmetries of an original Lagrangian $L$ (\ref{cmp31})
(Theorems \ref{y84} -- \ref{y41'}).

Let a vertical vector field $u=u^i\dr_i$ on $Y\to X$ be a
classical (Noether) symmetry of the Lagrangian $L$ (\ref{cmp31}),
i.e.,
\mar{m49}\beq
\bL_{J^1u}L=(u^i\dr_i +d_\la u^i\dr_i^\la)\cL\om=0. \label{m49}
\eeq
Since
\mar{y50}\beq
J^1u(y^i_\la-\G^i_\la)=\dr_ku^i(y^k_\la-\G^k_\la), \label{y50}
\eeq
one easily obtains from the equality (\ref{m49}) that
\mar{y45}\beq
u^k\dr_k a^{\la\m}_{ij}+ \dr_iu^k a^{\la\m}_{kj} +
a^{\la\m}_{ik}\dr_ju^k =0. \label{y45}
\eeq
It follows that the summands of the Lagrangian (\ref{cmp31}) are
invariant separately, i.e.,
\mar{y31}\beq
J^1u(a^{\la\m}_{ij}\cF^i_\la\cF^j_\m)=0, \qquad
J^1u(c')=u^k\dr_kc'=0. \label{y31}
\eeq
The equalities (\ref{mos018}), (\ref{y50}) and (\ref{y45}) give
the transformation law
\mar{b1}\beq
J^1u(a^{\la\m}_{ij}\cF^j_\m)=-\dr_i u^k a^{\la\m}_{kj}\cF^j_\m.
\label{b1}
\eeq
The  relations (\ref{N21}) and (\ref{y45}) lead to the equality
\mar{y53}\beq
a^{\la\mu}_{ij}[u^k\dr_k\si_0{}^{jn}_{\mu\al} -\dr_ku^j
\si_0{}^{kn}_{\mu\al} - \si_0{}^{jk}_{\mu\al}\dr_ku^n
]a^{\al\nu}_{nb}=0. \label{y53}
\eeq

Let us compare symmetries of the Lagrangian $L$ (\ref{cmp31}) and
the Lagrangian $L_N$ (\ref{bv2}). Given the Legendre map $\wh L$
(\ref{N13}) and the tangent morphism
\be
T\wh L: TJ^1Y\to TN_L, \qquad \dot p_A=(\dot y^i\dr_i +\dot
y^k_\nu\dr^\nu_k) (M_A{}^i_\la a^{\la\m}_{ij}\cF^j_\m),
\ee
let us consider the map
\mar{y79}\ben
&& T\wh L\circ J^1u: J^1Y\ni (x^\la,y^i,y^i_\la) \to
\label{y79}\\
&& u^i\dr_i + (u^k\dr_k +\dr_\nu u^k\dr^\nu_k) (M_A{}^i_\la
a^{\la\m}_{ij}\cF^j_\m)\dr^A= \nonumber \\
&& \qquad
u^i\dr_i + [u^k\dr_k(M_A{}^i_\la) a^{\la\m}_{ij}\cF^j_\m
+M_A{}^i_\la J^1u
(a^{\la\m}_{ij}\cF^j_\m)]\dr^A= \nonumber\\
&& \qquad  u^i\dr_i + [u^k\dr_k(M_A{}^i_\la)
a^{\la\m}_{ij}\cF^j_\m -M_A{}^i_\la \dr_iu^k
a^{\la\m}_{kj}\cF^j_\m]\dr^A=\nonumber \\
&& \qquad u^i\dr_i + [u^k\dr_k(a\circ\si_0)^{\m i}_{j\la}\cP^\la_i-
(a\circ\si_0)^{\m i}_{j\la}\dr_i u^k\cP^\la_k]\dr^j_\m\in TN_L,
\nonumber
\een
where the relations (\ref{y30}) and (\ref{b1}) have been used. Let
us assign to a point $(x^\la,y^i,\cP^\la_i)\in N_L$ some point
\mar{y81}\beq
(x^\la,y^i,y^i_\la) \in \wh L^{-1}(x^\la,y^i,\cP^\la_i)
\label{y81}
\eeq
and then the image of the point (\ref{y81}) under the morphism
(\ref{y79}). We obtain the map
\mar{y82}\beq
v_N:(x^\la,y^i,\cP^\la_i) \to u^i\dr_i +
[u^k\dr_k(a\circ\si_0)^{\m i}_{j\la}\cP^\la_i- (a\circ\si_0)^{\m
i}_{j\la}\dr_i u^k\cP^\la_k]\dr^j_\m \label{y82}
\eeq
which is independent of a choice of the point (\ref{y81}).
Therefore, it is a vector field on a Lagrangian constraint space
$N_L$. This vector field gives rise to a vector field
\mar{y83}\beq
Jv_N=u^i\dr_i + [u^k\dr_k(a\circ\si_0)^{\m i}_{j\la}\cP^\la_i-
(a\circ\si_0)^{\m i}_{j\la}\dr_i u^k\cP^\la_k]\dr^j_\m + d_\la
u^i\dr^\la_i \label{y83}
\eeq
on $N_L\op\times_Y J^1Y$.

\begin{theorem} \label{y84} \mar{y84}
The Lie derivative $\bL_{Jv_N} L_N$ of the Lagrangian $L_N$
(\ref{bv2}) along the vector field $Jv_N$ (\ref{y83}) vanishes,
i.e., any Noether symmetry $u$ of the Lagrangian $L$ (\ref{cmp31})
yields the symmetry $v_N$ (\ref{y82}) of the Lagrangian $L_N$
(\ref{bv2}).
\end{theorem}

\begin{proof}
One can show that
\mar{y87}\beq
v_N(\cP^\la_i)=-\dr_iu^k\cP^\la_k \label{y87}
\eeq
on the constraint space $\cR^\la_i=0$. Then the invariance
condition $Jv_N(\cL_N)=0$ falls into the three equalities
\mar{y85}\beq
Jv_N(\si_0{}^{ij}_{\la\m}\cP^\la_i\cP^\m_j)=0, \qquad
Jv_N(\cP^\la_i\cF^i_\la)=0, \qquad Jv_N(c')=0. \label{y85}
\eeq
The latter is exactly the second equality (\ref{y31}). The first
equality (\ref{y85}) is satisfied due to the relations (\ref{y53})
and (\ref{y87}). The second one takes a form
\mar{y90}\beq
Jv_N(\cP^\la_i(y^i_\la-\G^i_\la))=0. \label{y90}
\eeq
It holds owing to the relations (\ref{y50}) and (\ref{y87}).
\end{proof}

Turn now to symmetries of the Lagrangian $L_\Pi$ (\ref{m32}).
Since $L_\Pi$ is the pull-back of $L_N$ onto $\Pi\op\times_Y
J^1Y$, its symmetry must be an appropriate lift of the vector
field $v_N$ (\ref{y82}) onto $\Pi$.

Given a vertical vector field $u$ on $Y\to X$, let us consider its
canonical lift (\ref{Q4}):
\mar{y41}\beq
\wt u=u^i\dr_i - \dr_iu^j p^\la_j\dr^i_\la,  \label{y41}
\eeq
onto the Legendre bundle $\Pi$. It readily observed that the
vector field $\wt u$ is projected onto the vector field $v_N$
(\ref{y82}).

Let us additionally suppose that the one-parameter group of
automorphisms of $Y$ generated by $u$ preserves the splitting
(\ref{N18}), i.e., $u$ obeys the condition
\mar{y49'}\beq
u^k\dr_k(\si_0{}^{im}_{\la\nu}a^{\nu\m}_{mj})+
\si_0{}^{im}_{\la\nu}a^{\nu\m}_{mk}\dr_ju^k -
\dr_ku^i\si_0{}^{km}_{\la\nu}a^{\nu\m}_{mj} =0. \label{y49'}
\eeq
The relations (\ref{y50}) and (\ref{y49'}) lead to the
transformation law
\mar{m53}\beq
J^1u(\cF^i_\m)=\dr_ju^i\cF^j_\m.  \label{m53}
\eeq

\begin{theorem} \label{y41'} \mar{y41'}
If the condition (\ref{y49'}) holds, the vector field $\wt u$
(\ref{y41}) is a symmetry of the Lagrangian $L_\Pi$ (\ref{m32})
iff $u$ is a Noether symmetry of the original Lagrangian $L$
(\ref{cmp31}).
\end{theorem}

\begin{proof}
Due to the condition (\ref{y49'}), the vector field $\wt u$
(\ref{y41}) preserves the splitting (\ref{N20}), i.e.,
\mar{y55}\beq
\wt u(\cP^\la_i)=-\dr_i u^k\cP^\la_k, \qquad \wt
u(\cR^\la_i)=-\dr_i u^k\cR^\la_k.
\eeq
The vector field $\wt u$ gives rise to the vector field
(\ref{Q4'}):
\mar{y40}\beq
J\wt u=u^i\dr_i - \dr_iu^j p^\la_j\dr^i_\la +d_\la u^i\dr^\la_i,
\label{y40}
\eeq
on $\Pi\op\times_Y J^1Y$, and we obtain the Lagrangian symmetry
condition
\mar{m49'}\beq
(u^i\dr_i-\dr_j u^ip^\la_i\dr^j_\la +d_\la u^i\dr_i^\la)\cL_\Pi=0.
\label{m49'}
\eeq
It is readily observed that the first and third terms of a
Lagrangian $L_\Pi$ are separately invariant due to the relations
(\ref{y31}) and (\ref{m53}). Its second term is invariant owing to
the equality (\ref{y53}). Conversely, let the invariance condition
(\ref{m49'}) hold. It falls into the independent equalities
\mar{m50}\beq
J\wt u(\si_0{}^{ij}_{\la\m}p^\la_i p^\m_j) =0, \qquad
 J\wt u(p^\la_i\cF_\la^i)=0, \qquad
u^i\dr_i c'=0, \label{m50}
\eeq
i.e., the Lagrangian $L_\Pi$ is invariant iff its three summands
are separately invariant. One obtains at once from the second
condition (\ref{m50}) that the quantity $\cF$ is transformed as
the dual of momenta $p$. Then the first condition (\ref{m50})
shows that the quantity $\si_0p$ is transformed by the same law as
$\cF$. It follows that the term $a\cF\cF$ in the Lagrangian $L$
(\ref{cmp31}) is transformed as $a(\si_0 p)(\si_0 p)=\si_0pp$,
i.e., it is invariant. Then this Lagrangian is invariant due to
the third equality (\ref{m50}).
\end{proof}

\begin{remark} \label{ps77} \mar{ps77}
At the same time, a Lagrangian $L_\Pi$ may possess additional
non-classical symmetries which do not come from symmetries of an
original Lagrangian $L$. For instance, let us assume that $Y\to X$
is an affine bundle modelled over a vector bundle $\ol Y\to X$. In
this case, the Legendre bundle $\Pi$ (\ref{00}) is isomorphic to
the product
\be
\Pi=Y\op\times_X (\ol Y^*\op\ot_X\op\w^n T^*X\op\ot_X TX)
\ee
such that transition functions of coordinates $p^\la_i$ are
independent of $y^i$. Then the splitting (\ref{N20}) takes a form
\mar{y26}\beq
\Pi=Y\op\times_X(\ol{\Ker \si_0}\op\oplus_X\ol N_L), \label{y26}
\eeq
where $\ol{\Ker \si_0}$ and $\ol N_L$ are fibre bundles over $X$
such that
\be
\Ker \si_0=\pi^*\ol{\Ker \si_0},
\ee
and $N_L=\pi^*\ol N_L$ are their pull-backs onto $Y$. The
splitting (\ref{y26}) keeps the coordinate form (\ref{N20'}). The
splittings (\ref{N18}) and (\ref{y26}) lead to the decomposition
\mar{y25}\beq
\Pi\op\times_Y J^1Y=(\ol{\Ker \si_0}\op\oplus_X\ol N_L)\op\times_Y
(\Ker\wh L\op\oplus_Y {\rm Im}(\si_0\circ\wh L)). \label{y25}
\eeq
In view of this decomposition, let us associate to any section
$\xi$ of $\ol{\Ker \si_0}\to X$ the vector field
\mar{y27}\beq
u_\Pi=\xi_a (M^{-1})^a{}_i^\la\dr^i_\la, \label{y27}
\eeq
on $\Pi$. Its lift (\ref{Q4'}) onto $\Pi\op\times_Y J^1Y$ keeps
the coordinate form
\mar{y27'}\beq
Ju_\Pi=\xi_a (M^{-1})^a{}_i^\la\dr^i_\la. \label{y27'}
\eeq
It is readily observed that the Lie derivative of the Lagrangian
$L_\Pi$ (\ref{m32}) along the vector field (\ref{y27'}) vanishes,
i.e., $u_\Pi$ is a symmetry of $L_\Pi$. Moreover, the vector
fields (\ref{y27}), parameterized by sections $\xi$ of $\ol{\Ker
\si_0}\to X$, is a gauge symmetry of $L_\Pi$
\cite{book09,gauge09}. However, it does not come from symmetries
of an original Lagrangian $L$.
\end{remark}

\section{PS Hamiltonian gauge theory}

Yang -- Mills gauge theory of principal connections provides the
most physically relevant example of a quadratic Lagrangian and PS
Hamiltonian systems \cite{book,book09}. The peculiarity of gauge
theory lies in the fact that the splittings (\ref{N18}) and
(\ref{N20}) of its configuration and momentum phase spaces are
canonical.

Let $P\to X$ be a principal bundle with a structure Lie group $G$.
Being $G$-equivariant, principal connections on $P\to X$ are
represented by sections of the affine bundle
\mar{br3}\beq
C=J^1P/G\to X, \label{br3}
\eeq
called the bundle of principal connections. It is modelled over a
vector bundle $T^*X\ot V_GP$, where $V_GP=VP/G\to X$ is the fibre
bundle in Lie algebras $\cG$ of the group $G$. Given a basis
$\{\ve_r\}$ for $\cG$, we obtain local fibre bases $\{e_r\}$ for
$V_GP$. The connection bundle $C$ (\ref{br3}) is coordinated by
$(x^\m,a^r_\m)$ such that, written relative to these coordinates,
sections $A=A^r_\m dx^\m\ot e_r$ of $C\to X$ are familiar local
connection one-forms, regarded as gauge potentials.

There is one-to-one correspondence between the sections $\xi=\xi^r
e_r$ of $V_GP\to X$ and the vector fields on $P$ which are
infinitesimal generators of one-parameter groups of vertical
automorphisms (gauge transformations) of $P$. Any section $\xi$ of
$V_GP\to X$ yields the vector field
\mar{br6}\beq
u_\xi=u^r_\m \dr_r^\m=(\dr_\m\xi^r + c^r_{pq}a^p_\m\xi^q)\dr_r^\m
\label{br6}
\eeq
on $C$, where $c^r_{pq}$ are the structure constants of a Lie
algebra $\cG$.

The configuration space of gauge theory is the first order jet
manifold $J^1C$ equipped with the adapted coordinates
$(x^\la,a^m_\la,a^m_{\m\la})$. This configuration space admits the
canonical splitting
\mar{N31'}\ben
&& J^1C= C_+\op\oplus_C  C_-=C_+\op\oplus (C\op\times_X\op\w^2T^*X\op\ot_X V_GP),
\label{N31'} \\
&&  a^r_{\m\la}=\frac12(a^r_{\m\la}+a^r_{\la\m}-c^r_{pq}a^p_\m
a^q_\la) +\frac12(a^r_{\m\la}-a^r_{\la\m} +c^r_{pq}a^p_\m
a^q_\la), \nonumber
\een
with the corresponding projections
\mar{b4127,8}\ben
&&{\cal S}: J^1 C\to C_+, \qquad
{\cal S}^r_{\m\la}= a^r_{\m\la}+a^r_{\la\m}-c^r_{pq}a^p_\m
a^q_\la, \label{b4127}\\
&& \cF: J^1 C\to C_-,\qquad
 \cF^r_{\m\la}= a^r_{\m\la}-a^r_{\la\m} +c^r_{pq}a^p_\m a^q_\la,
\label{b4128}
\een
where $\cF$  is the strength of gauge fields.

Gauge theory of principal connections on $P\to X$ is characterized
by almost regular first order Yang--Mills Lagrangian
\mar{5.1'}\beq
L_{\rm YM}=\frac{1}{4}a^G_{pq}g^{\la\m}g^{\bt\n}\cF^p_{\la
\beta}\cF^q_{\m\n}\sqrt{\nm g}\,\om, \qquad  g=\det(g_{\m\nu}), \label{5.1'}
\eeq
on $J^1C$, where  $a^G$ is a non-degenerate $G$-invariant metric
on the Lie algebra $\cG_r$ and $g$ is a non-degenerate world
metric on $X$. It possesses the gauge symmetries $u_\x$
(\ref{br6}). Their jet prolongation onto $J^1C$ read
\mar{ps50}\beq
J^1u_\xi=(\dr_\mu\x^r+c^r_{qp}a^q_\mu\x^p)\dr^\mu_r +
(c^r_{pq}(a^p_{\la\m}\xi^q +a^p_\m\dr_\la\xi^q)
+\dr_\la\dr_\m\xi^r)\dr_r^{\la\m}, \label{ps50}
\eeq
and we have transformation laws
\mar{ps60}\ben
&& J^1u_\xi(\cF^r_{\la\m})=c^r_{pq}\cF^p_{\la\m}\xi^q,\nonumber \\
&& J^1u_\xi(\cS^r_{\la\m})= c^r_{pq}\cS^p_{\la\m}\xi^q
+c^r_{pq}a^p_\m\dr_\la\xi^q+ \dr_\la\dr_\m\xi^r. \label{ps60}
\een

The Euler -- Lagrange operator of the Yang -- Mills Lagrangian
$L_{\rm YM}$ (\ref{5.1'}) is
\be
\dl L_{\rm YM}=\cE_{\rm YM}=\cE^\m_r\thh^\m_r\w\om=(\dl^n_rd_\la
+c^n_{rp}a^p_\la)(a^G_{nq}g^{\m\al}g^{\la\bt}
\cF^q_{\al\bt}\sqrt{|g|})\thh_\m^r\w\om.
\ee
Its kernel defines the Yang -- Mills equation
\mar{57f14}\beq
\cE^\m_r= (\dl^n_rd_\la
+c^n_{rp}a^p_\la)(a^G_{nq}g^{\m\al}g^{\la\bt}
\cF^q_{\al\bt}\sqrt{|g|})=0. \label{57f14}
\eeq

A momentum phase space of gauge theory is the Legendre bundle
\mar{b4129}\beq
\pi_{\Pi C}: \Pi\to C, \qquad \Pi =\op\wedge^n T^*X\op\otimes_C
TX\op\otimes_C [C\times\ol C]^*, \label{b4129}
\eeq
endowed with holonomic coordinates $(x^\la,a^p_\la,p^{\mu\la}_m)$.
The Legendre bundle $\Pi$ (\ref{b4129}) admits the canonical
decomposition (\ref{N20}):
\ben
&& \Pi=
\Pi_+\op\oplus_C\Pi_-,\label{N32} \\
&&  p^{\mu\lambda}_m= \cR^{(\mu\lambda)}_m +
\cP^{[\mu\lambda]}_m= p^{(\mu\lambda)}_m +
p^{[\mu\lambda]}_m=\frac{1}{2}(p^{\mu\lambda}_m+ p^{\lambda\mu}_m)
+ \frac{1}{2}(p^{\mu\lambda}_m- p^{\lambda\mu}_m). \nonumber
\een

The Legendre map associated to the Lagrangian (\ref{5.1'}) takes a
form
 \ben
 &&p^{(\mu\la)}_m\circ\wh L_{YM}=0, \label{5.2a}\\
&&p^{[\mu\la]}_m\circ\wh
L_{YM}=a^G_{mn}g^{\m\al}g^{\la\bt} \cF^n_{\al\bt}\sqrt{\mid
g\mid}. \label{5.2b}
\een
A glance at this morphism shows that $\Ker\wh L_{YM}=C_+$, and the
Lagrangian constraint space is
\mar{ps66}\beq
N_L=\wh L_{YM}(J^1C)=\Pi_-. \label{ps66}
\eeq
It is defined by the equation $p^{(\mu\la)}_m=0$ (\ref{5.2a}).
Obviously, $N_L$ is an imbedded submanifold of $\Pi$, and the
Lagrangian $L_{\rm YM}$ is almost regular.

Let us consider connections $\G$ on a fibre bundle $C\to X$ which
take their values into $\Ker\wh L$, i.e.,
\mar{69}\beq
 \G:C\to C_+, \qquad
\G^r_{\la\m}-\G^r_{\m\la}+c^r_{pq}a^p_\la a^q_\m=0. \label{69}
\eeq
Given a symmetric linear connection $K$ (\ref{08}) on $T^*X$,
every principal connection $B$ on a principal bundle $P\to X$
gives rise to a connection $\G_B: C\to C_+$ such that
\be
\G_B\circ B={\cal S}\circ J^1B.
\ee
It reads
\beq
\G_B{}^r_{\la\m}=\frac{1}{2} [\dr_\mu B^r_\la+\dr_\la B^r_\mu
-c^r_{pq}a^p_\la a^q_\mu  +  c^r_{pq} (a^p_\la B^q_\m +a^p_\m
B^q_\la)] - K_\la{}^\bt{}_\mu(a^r_\bt-B^r_\bt). \label{3.7}
\eeq

Given the connection (\ref{3.7}), the corresponding Hamiltonian
form (\ref{N22}):
\ben
&&H_B=p^{\la\m}_r d a^r_\m\w\om_\la-
p^{\la\m}_r\G_B{}^r_{\la\m}\om-\wt{\cH}_{YM}\om, \label{5.3}\\
&&\wt{\cH}_{YM}= \frac{1}{4}a^{mn}_Gg_{\mu\nu}
g_{\lambda\beta} p^{[\mu\lambda]}_m p^{[\nu\beta]}_n\sqrt{\nm g},
\nonumber
\een
is associated to the Lagrangian $L_{\rm YM}$ (\ref{5.1'}). It is
the Poincar\'e -- Cartan form of a Lagrangian
\mar{yyy}\beq
L_H=[p^{\la\m}_r (a^r_{\la\m}- \G_B{}^r_{\la\m})
-\wt{\cH}_{YM}]\om \label{yyy}
\eeq
on $\Pi\op\times_C J^1C$. The pull-back of any Hamiltonian form
$H_B$ (\ref{5.3}) onto the Lagrangian constraint space $N_L$
(\ref{ps66}) is the constrained Hamiltonian form (\ref{b4300}):
\mar{b4130}\beq
H_N=i^*_NH_B=p^{[\la\m]}_r(da^r_\mu\w\om_\la + \frac12
c^r_{pq}a^p_\la a^q_\mu\om) -\wt{\cH}_{YM}\om. \label{b4130}
\eeq
The corresponding constrained Lagrangian $L_N$ on
\mar{ps67}\beq
N_L\op\times_C J^1C= \Pi_-\op\ot_C J^1C \label{ps67}
\eeq
reads
\mar{yyy1}\beq
L_N= (p^{[\la\m]}_r\cF^r_{\la\mu} -\wt{\cH}_{YM})\om. \label{yyy1}
\eeq
Its pull-back $L_\Pi$ onto $\Pi\op\times_C J^1C$ is
\mar{yyy10}\beq
L_\Pi= (p^{\la\m}_r\cF^r_{\la\mu} -\wt{\cH}_{YM})\om.
\label{yyy10}
\eeq

Note that, in contrast with the Lagrangian (\ref{yyy}), the
constrained Lagrangian $L_N$ (\ref{yyy1}) possesses gauge
symmetries as follows. Gauge symmetries $u_\xi$ (\ref{br6}) of the
Yang -- Mills Lagrangian give rise to vector fields (\ref{y41}):
\mar{ps51}\beq
\wt u_\xi=(\dr_\mu\x^r+c^r_{qp}a^q_\mu\x^p)\dr^\mu_r -
c^r_{qp}\x^p p^{\la\m}_r\dr^q_{\la\m} \label{ps51}
\eeq
on $\Pi$. Vector fields $J^1u_\x$ (\ref{ps50}) and  $\wt u_\xi$
(\ref{ps51}) provide gauge symmetries
\mar{ps68}\beq
J\wt u_\x=J^1u +\wt u \label{ps68}
\eeq
of the Lagrangians $L_N$ (\ref{yyy1}) and $L_\Pi$ (\ref{yyy10}) in
accordance with Theorems \ref{y84} -- \ref{y41'}.

The Hamiltonian form $H_B$ (\ref{5.3}) yields the covariant
Hamilton equation which  consist of the equation (\ref{5.2b}) and
the equations
\ben
&& a^m_{\la\mu}+
a^m_{\m\la}=2\G_B{}^m_{(\la\mu)},\label{5.6}\\
&& p^{\la\mu}_{\la r}=c^q_{pr}r^p_\la
p^{[\la\m]}_q - c^q_{rp}B^p_\la p^{(\la\mu)}_q +K_\la{}^\mu{}_\nu
p^{(\la\nu)}_r. \label{5.5}
\een
The Hamilton equations (\ref{5.6}) and (\ref{5.2b}) are similar to
the equations (\ref{N23}) and (\ref{N28}), respectively. The
Hamilton equations (\ref{5.2b}) and (\ref{5.5}) restricted to the
Lagrangian constraint space (\ref{5.2a}) are precisely the
constrained Hamilton equation (\ref{N44}) for the constrained
Hamiltonian form $H_N$ (\ref{b4130}), and they are equivalent to
the Yang -- Mills equation (\ref{57f14}) for gauge potentials
$A=\pi_{\Pi C}\circ r$.

Different Hamiltonian forms $H_B$ lead to different equations
(\ref{5.6}). This equation is independent of momenta and, thus, it
exemplifies the gauge-type condition (\ref{N23}):
\be
\G_B{}\circ A={\cal S}\circ J^1A.
\ee
A glance at this condition shows that, given a solution $A$ of the
Yang -- Mills equation, there always exists a Hamiltonian form
$H_B$ (e.g., $H_{B=A}$) which obeys the condition (\ref{2.36}),
i.e.,
 \be
\wh H_B\circ\wh L_{YM}\circ J^1A=J^1A.
\ee
Consequently, the Hamiltonian forms $H_B$ (\ref{5.3})
parameterized by principal connections $B$ constitute a complete
set.

It should be emphasized that the gauge-type condition (\ref{5.6})
differs from the familiar gauge conditions in gauge theory which
single out a representative of each gauge coset (with the accuracy
to Gribov's ambiguity). Namely, if a gauge potential $A$ is a
solution of the Yang -- Mills equation, there exists a gauge
conjugate potential $A'$ which also is a solution of the Yang --
Mills equation and satisfies a given gauge condition. At the same
time, not every solution of the Yang -- Mills equation is a
solution of the system of the Yang -- Mills equation and a certain
gauge condition. In other words, there are solutions of the  Yang
-- Mills equation which are not singled out by the gauge
conditions known in gauge theory. In this sense, this set of gauge
conditions is not complete. In gauge theory, this lack is not
essential since one can think of all gauge conjugate potentials as
being physically equivalent, but not in the case of other
non-regular Lagrangian systems, e.g., that of Proca fields
(\cite{book}, Example 4.6.5).

In the framework of the PS Hamiltonian description of quadratic
Lagrangian systems, there is a complete set of gauge-type
conditions in the sense that, for any solution of the Euler --
Lagrange equation, there exist constrained Hamilton equation
equivalent to this Euler -- Lagrange equation and a supplementary
gauge-type condition which this solution satisfies.

In gauge theory where gauge conjugate solutions are treated
physically equivalent, one may replace  the equation (\ref{5.6})
with a condition on the quantity
\be
({\cal S}\circ J^1A)^r_{\la\m} =\frac12 (\dr_\la A^r_\mu +\dr_\mu
A^r_\la -c^r_{pq}A^p_\la A^q_\mu),
\ee
which supplements the Yang -- Mills equation and plays a role of a
gauge condition due to the gauge transformation law (\ref{ps60}).
In particular,
\mar{ps70}\beq
g^{\la\m}({\cal S}\circ J^1A)^r_{\la\m} =\al^r(x) \label{ps70}
\eeq
recovers the familiar generalized Lorentz gauge condition.

\section{Affine Lagrangian and Hamiltonian systems}

Let us turn now to a case of an affine Lagrangian system on a
fibre bundle $Y\to X$ whose Lagrangian is given by the coordinate
expression
\beq
L=\cL\om, \qquad \cL=b^\la_i y^i_\la + c, \label{N24}
\eeq
where $b$ and $c$ are local functions on $Y$. The corresponding
Legendre map $\wh L$ (\ref{b330}) takes a form
\beq
p^\la_i\circ\wh L = b^\la_i. \label{N25}
\eeq
We have the commutative diagram
\be
&& \begin{array}{rcl}
 {J^1Y}  & \op\longrightarrow^{\wh L} &  {Q\subset \Pi}  \\
  & \searrow  \nearrow & _{b}\\
 & {Y} &
\end{array}, \\
&& b= b^\la_i\om_\la\ot dy^i,
\ee
where $Q=b(Y)$ is the image of a  section $b$ of a Legendre bundle
$\Pi\to Y$. Clearly, the Lagrangian (\ref{N24}) is almost regular
without fail.

Let $\G$ be an arbitrary connection (\ref{ps9}) on a fibre bundle
$Y\to X$, and let $\wh\G$ the associated Hamiltonian map
(\ref{b420}). This Hamiltonian map satisfies the condition
(\ref{ps130}), where $\wh L$ is the Legendre morphism (\ref{N25}).
Let us consider the Hamiltonian form (\ref{ps131}) corresponding
to $\wh\G$. It reads
\beq
H=H_\G+L\circ \G=p^\la_idy^i\w\om_\la - (p^\la_i
-b^\la_i)\G_\la^i\om + c\om, \label{N26}
\eeq
and is weakly associated to the affine Lagrangian (\ref{N24}). The
corresponding Hamiltonian map
\beq
y^i_\la\circ\wh H = \G^i_\la \label{N27}
\eeq
coincides with $\wh\G$, i.e., $H$ (\ref{N26}) is associated to
$L$.

The Hamiltonian form $H$ (\ref{N26}) is affine in canonical
momenta. It follows that the Hamilton equation (\ref{b4100a}) for
$H$ reduce to the gauge-type condition
\be
\dr_\la r^i=\G^i_\la,
\ee
whose solutions are integral sections of the connection $\G$.

Conversely, for each section $s$ of a fibre bundle $Y\to X$, there
exists a connection $\G$ on $Y$ whose integral section is $s$.
Then, the corresponding Hamiltonian map (\ref{N27}) obeys the
condition (\ref{2.36}). It follows that the Hamiltonian forms
(\ref{N26}) parameterized by connections $\G$ on a fibre bundle
$Y\to X$ constitute a complete family.

The most physically relevant examples of affine Lagrangian and PS
Hamiltonian systems are Dirac spinor fields and metric
affine-gravitation theory \cite{book,sard94,book95}.

\section{Functional integral quantization}

The fact that PS Hamiltonian system with the Hamiltonian form $H$
(\ref{b418}) on a Legendre bundle $\Pi$ is equivalent to a first
order Lagrangian system on $\Pi$ with the Lagrangian $L_H$
(\ref{Q3}) enables us to quantize this PS Hamiltonian system in
the framework of familiar perturbative quantum field theory
\cite{bs04,quant94}.

If there is no constraints and the matrix $\dr^2\cH/\dr p^\m_i\dr
p^\nu_j$ is non-degenerate and positive-definite, this
quantization is given by the generating functional
\mar{m2}\beq
Z=\cN^{-1}\int\exp\{\int(\cL_\cH +\Lambda + iJ_iy^i+iJ^i_\m
p^\m_i) dx \}\op\prod_x [dp(x)][dy(x)] \label{m2}
\eeq
of Euclidean Green functions, where $\Lambda$ comes from the
normalization condition
\be
\int \exp\{\int(-\frac12\dr_\m^i\dr_\nu^j\cH p^\m_i
p^\nu_j+\La)dx\}\op\prod_x[dp(x)]=1.
\ee
If a Hamiltonian $\cH$ is degenerate, the Lagrangian $L_H$
(\ref{Q3}) may admit gauge symmetries. In this case, integration
of a generating functional along gauge group orbits must be
finite. If there are constraints, the Lagrangian system with the
Lagrangian $L_H$ (\ref{Q3}) restricted to a constraint space is
quantized.

In order to verify this functional integral quantization scheme,
we apply it to PS Hamiltonian systems associated to Lagrangian
systems with quadratic Lagrangians (\ref{N12}). Note that, in the
framework of perturbative quantum field theory, any Lagrangian is
split into a sum of some quadratic Lagrangian (\ref{N12}) and an
interaction term quantized as a perturbation.

For instance, let the Lagrangian (\ref{N12}) be hyperregular,
i.e., the matrix function $a$ is non-degenerate. Then there exists
a unique associated Hamiltonian system whose associated
Hamiltonian form $H$ (\ref{cc311}) is quadratic in momenta
$p^\m_i$, and so is the corresponding Lagrangian $L_H$ (\ref{Q3}).
If the matrix function $a$ is positive-definite on an Euclidean
space-time, the generating functional (\ref{m2}) is a Gaussian
integral of momenta $p^\m_i(x)$. Integrating $Z$ with respect to
$p^\m_i(x)$, one restarts the generating functional of quantum
field theory with the original Lagrangian (\ref{N12}). We extend
this result to theories with almost regular Lagrangians $L$
(\ref{N12}), e.g., Yang -- Mills gauge theory. The key point is
that, though such a Lagrangian $L$ yields Lagrangian constraints
$N_L$ (\ref{zzz}) and admits different associated Hamiltonian
forms $H$, all the Lagrangians $L_H$ coincide on a constraint
space $J^1Y\times_Y N_L$, and we have a unique constrained
Lagrangian system with a Lagrangian $L_N$ which is equivalent to
the original one.

Let us quantize a Lagrangian system with the Lagrangian
$L_N=\cL_N\om$ (\ref{bv2}) on a product $J^1Y\times_Y N_L$. In the
framework of a perturbative quantum field theory, we should assume
that $X=\mathbb R^n$ and $Y\to X$ is a trivial affine bundle. It
follows that both the original coordinates $(x^\la,y^i,p^\la_i)$
and the adapted coordinates $(x^\la,y^i,p_a,p_A)$ on the Legendre
bundle $\Pi$ are global. Passing to field theory on an Euclidean
space $\mathbb R^n$, we also assume that the matrix $a$ in the
Lagrangian $L$ (\ref{cmp31}) is positive-definite, i.e.,
$a_{AA}>0$.

Let us start with the Lagrangian (\ref{bv2}) without gauge
symmetries. Since a Lagrangian constraint space $N_L$ can be
equipped with the adapted coordinates $p_A$, the generating
functional of Euclidean Green functions of a Lagrangian system in
question reads
\mar{m43}\beq
Z=\cN^{-1}\int\exp\{\int (\cL_N  +\frac12{\rm tr}\,\ln \ol \si_0+
iJ_iy^i+iJ^Ap_A)\om\}\op\prod_x [dp_A(x)][dy(x)], \label{m43}
\eeq
where $\cL_N$ with respect to the adapted coordinates is given by
an expression
\be
\cL_N=M^{-1}{}_i^{\la A}p_A\cF_\la^i
-\frac12\op\sum_A(a_{AA})^{-1}(p_A)^2 + c',
\ee
and $\ol \si_0$ is a square matrix
\be
\ol \si_0^{AB}=M^{-1}{}_i^{\la A} M^{-1}{}_j^{\m B}
\si_0{}^{ij}_{\la\m} =\dl^{AB}(a_{AA})^{-1}.
\ee

The generating functional (\ref{m43}) is a Gaussian integral of
functional variables $p_A(x)$. Its integration with respect to
$p_A(x)$ under the condition $J^A=0$ restarts a generating
functional
\mar{m44}\beq
Z=\cN^{-1}\int\exp\{\int (\cL + iJ_iy^i)\om\}\op\prod_x [dy(x)],
\label{m44}
\eeq
of the original Lagrangian field system on $Y$ with the Lagrangian
(\ref{cmp31}). However, the generating functional (\ref{m43}) can
not be rewritten with respect to the original variables $p^\m_i$,
unless $a$ is a non-degenerate matrix function.

In order to overcome this difficulty, let us consider a Lagrangian
system on the whole Legendre manifold $\Pi$ with the Lagrangian
$L_\Pi$ (\ref{m32}). Since this Lagrangian is constant along the
fibres of a vector bundle $\Pi\to N_L$, the integration of the
generating functional of this field model with respect to
variables $p_a(x)$ should be finite. One can choose the generating
functional in a form
\mar{m21}\beq
Z=\cN^{-1}\int\exp\{\int
(\cL_\Pi-\frac12\si_1{}^{ij}_{\la\m}p^\la_ip^\m_j
 +\frac12{\rm tr}\,\ln \si +
iJ_iy^i+iJ^i_\m p_i^\m)\om\} \op\prod_x [dp(x)][dy(x)].
\label{m21}
\eeq
Its integration with respect to momenta $p_i^\la(x)$ restarts the
generating functional (\ref{m44}) of the original Lagrangian
system on $Y$.

\begin{remark}
Strictly speaking, since a Lagrangian $L_\Pi$ may possess gauge
symmetries (Remark \ref{ps77}), in order to obtain the generating
functional (\ref{m21}), one can follow a procedure of quantization
of gauge-invariant Lagrangian systems. In a case of the Lagrangian
$L_\Pi$ (\ref{m32}), this procedure is rather trivial because the
space of momenta variables $p_a(x)$ coincides with the translation
subgroup of the gauge group Aut$\,\Ker\si_0$.
\end{remark}

Now let us suppose that the Lagrangian $L_N$ (\ref{bv2}) and,
consequently, the Lagrangian $L_\Pi$ (\ref{m32}) (Theorem
\ref{y41'}) are invariant under some gauge group $G_X$ of vertical
automorphisms of a fibre bundle $Y\to X$ (and the induced
automorphisms of $\Pi\to X$) which acts freely on a space of
sections of $Y\to X$. Its infinitesimal generators are represented
by vertical vector fields $u=u^i(x^\m,y^j)\dr_i$ on $Y\to X$ which
give rise to the vector fields $\ol u=J\wt u$ (\ref{y40}):
\mar{m47}\beq
\ol u=u^i\dr_i-\dr_j u^ip^\la_i\dr^j_\la +d_\la u^i\dr_i^\la,
\qquad d_\la=\dr_\la +y_\la^i\dr_i, \label{m47}
\eeq
on $\Pi\op\times_Y J^1Y$. Let us also assume that $G_X$ is indexed
by $m$ parameter functions $\xi^r(x)$ such that
$u=u^i(x^\la,y^j,\xi^r)\dr_i$, where
\mar{m48}\beq
u^i(x^\la,y^j,\xi^r)=u_r^i(x^\la,y^j)\xi^r
+u_r^{i\m}(x^\la,y^j)\dr_\m\xi^r \label{m48}
\eeq
are linear first order differential operators on a space of
parameters $\xi^r(x)$ \cite{book09,gauge09}. The vector fields
$u(\xi^r)$ must satisfy the commutation relations
\be
[u(\xi^q),u(\xi'^p)]=u(c^r_{pq}\xi'^p\xi^q),
\ee
where $c^r_{pq}$ are structure constants. The Lagrangian $L_\Pi$
(\ref{m32}) is invariant under the above mentioned gauge
transformations iff its Lie derivative $\bL_{\ol u}L_\Pi$ along
vector fields (\ref{m47}) vanishes, i.e.,
\mar{m49s}\beq
(u^i\dr_i-\dr_j u^ip^\la_i\dr^j_\la +d_\la u^i\dr_i^\la)\cL_\Pi=0.
\label{m49s}
\eeq
Since the operator $\bL_{\ol u}$ is linear in momenta $p^\m_i$,
the condition (\ref{m49s}) falls into the independent conditions
(\ref{m50}). It follows that a Lagrangian $L_\Pi$ is
gauge-invariant iff its three summands are separately
gauge-invariant.

Since $\cS^i_\la=y^i_\la-\cF^i_\la$, one can easily derive from
the formula (\ref{m53}) the transformation law
\mar{m54}\beq
\ol u(\cS^i_\m)=d_\m u^i-\dr_ju^i\cF^j_\m= d_\m
u^i-\dr_ju^i(y^j_\m-\cS^j_\m) =\dr_\m u^i +\dr_ju^i\cS^j_\m
\label{m54}
\eeq
of $\cS$. A glance at this expression shows that the gauge group
$G_X$ acts freely on a space of sections $\cS(x)$ of the fibre
bundle $\Ker\wh L\to Y$ in the splitting (\ref{N18}). Let the
number $m$ of parameters of a gauge group $G_X$ do not exceed the
fibre dimension of $\Ker\wh L\to Y$. Then some combinations
$b^r{}_i^\m\cS^i_\m$ of $\cS^i_\m$ can be used as the gauge
condition
\be
b^r{}_i^\m\cS^i_\m(x)-\al^r(x)=0,
\ee
similar to the generalized Lorentz gauge (\ref{ps70}) in
Yang--Mills gauge theory.

Turn now to quantization of a Lagrangian system with the
gauge-invariant Lagrangian $L_\Pi$ (\ref{m32}). In accordance with
the well-known quantization procedure, let us modify the
generating functional (\ref{m21}) as follows
\mar{m55}\ben
&& Z=\cN^{-1}\int\exp\{\int (\cL_\Pi-\frac12\si_1{}^{ij}_{\la\m}p^\la_ip^\m_j
 +\frac12{\rm tr}\,\ln \si -\frac12 h_{rs}\al^r\al^s +
iJ_iy^i+iJ^i_\m p_i^\m)\om\} \nonumber \\
&& \qquad \Delta\op\prod_x
\op\times^r\dl(b^r{}_i^\m\cS^i_\m(x)-\al^r(x))[d\al(x)]
[dp(x)][dy(x)]=
\nonumber\\
&& \cN'^{-1}\int\exp\{\int (\cL_\Pi-\frac12\si_1{}^{ij}_{\la\m}p^\la_ip^\m_j
 +\frac12{\rm tr}\,\ln \si -\frac12 h_{rs}b^r{}_i^\m b^s{}_j^\la\cS^i_\m
\cS^j_\la +
iJ_iy^i+iJ^i_\m p_i^\m)\om\} \nonumber \\
&& \qquad \Delta\op\prod_x [dp(x)][dy(x)],
\label{m55}
\een
where
\be
\int\exp\{\int (-\frac12 h_{rs}\al^r\al^s)\om\} \op\prod_x
[d\al(x)]
\ee
is a Gaussian integral, and the factor $\Delta$ is defined by the
condition
\be
\Delta\int \op\prod_x
\op\times^r\dl(u(\xi)(b^r{}_i^\m\cS^i_\m))[d\xi(x)]=1.
\ee
We have the linear second order differential operator
\mar{m56}\beq
M^r_s\xi^s= u(\xi)(b^r{}_i^\m\cS^i_\m(x))=b^r{}_i^\m(\dr_\m
u^i(\xi) + \dr_ju^i(\xi)\cS^j_\m) \label{m56}
\eeq
on the parameter functions $\xi(x)$, and obtain $\Delta=\det M$.
Then the generating functional (\ref{m55}) takes a form
\mar{m57}\ben
&& Z=\cN'^{-1}\int\exp\{\int (\cL_\Pi-\frac12\si_1{}^{ij}_{\la\m}p^\la_ip^\m_j
 +\frac12{\rm tr}\,\ln \si -\frac12 h_{rs}b^r{}_i^\m b^s{}_j^\la\cS^i_\m
\cS^j_\la -\ol c_r M^r_sc^s + \nonumber \\
&& \qquad iJ_iy^i+iJ^i_\m p_i^\m)\om\} \op\prod_x [d\ol c][d c][dp(x)][dy(x)],
\label{m57}
\een
where $\ol c_r$, $c^s$ are odd ghost fields. Integrating $Z$
(\ref{m57}) with respect to momenta under the condition
$J^i_\m=0$, we come to the generating functional
\mar{m58}\beq
Z=\cN'^{-1}\int\exp\{\int (\cL-\frac12 h_{rs}b^r{}_i^\m
b^s{}_j^\la\cS^i_\m \cS^j_\la -\ol c_r M^r_sc^s+iJ_iy^i)\om\}
\op\prod_x [d\ol c][d c][dy(x)] \label{m58}
\eeq
of the original Lagrangian system on $Y$ with the gauge-invariant
Lagrangian $L$ (\ref{cmp31}).

Note that the Lagrangian
\mar{m59}\beq
\cL'=\cL-\frac12 h_{rs}b^r{}_i^\m b^s{}_j^\la\cS^i_\m \cS^j_\la
-\ol c_r M^r_sc^s \label{m59}
\eeq
fails to be gauge-invariant, but it admits the BRST symmetry whose
odd operator reads
\mar{m60}\ben
&& \vt=u^i(x^\m,y^i,c^s)\dr_i +d_\la u^i(x^\m,y^i,c^s)\dr_i^\la +\ol
v_r(x^\m,y^i,y^i_\m)\frac{\dr}{\dr\ol c_r} +  \nonumber\\
&& \qquad v^r(x^\m,y^i,c^s)
\frac{\dr}{\dr c^r}  +d_\la v^r(x^\m,y^i,c^s) \frac{\dr}{\dr
c^r_\la} +d_\m d_\la v^r(x^\m,y^i,c^s) \frac{\dr}{\dr
c^r_{\m\la}},
 \label{m60}\\
&& d_\la=\dr_\la +y^i_\la\dr_i + y^i_{\la\m}\dr_i^\m +
c^r_\la \frac{\dr}{\dr c^r} +c_{\la\m}^r\frac{\dr}{\dr c^r_\m}.
\nonumber
\een
Its components $u^i(x^\m,y^i,c^s)$ are given by the expression
(\ref{m48}) where parameter functions $\xi^r(x)$ are replaced with
the ghosts $c^r$. The components $\ol v_r$ and $v^r$ of the BRST
operator $\vt$ can be derived from the condition
\be
\vt(\cL')= - h_{rs} M^r_qb^s{}_j^\la \cS^j_\la c^q -\ol v_rM^r_q
c^q +\ol c_r \vt(\vt(b^r{}_j^\la \cS^j_\la))=0
\ee
of the BRST invariance of $\cL'$. This condition falls into the
two independent relations
\be
&&  h_{rs} M^r_qb^s{}_j^\la \cS^j_\la +\ol v_rM^r_q =0,\\
&& \vt(c^q)(\vt(c^p)(b^r{}_j^\la \cS^j_\la)) =u(c^p)(u(c^q)
(b^r{}_j^\la \cS^j_\la)) +u(v^r)(b^r{}_j^\la \cS^j_\la)=
u(\frac12c^r_{pq}c^pc^q+v^r)(b^r{}_j^\la \cS^j_\la)=0.
\ee
Hence, we obtain
\be
\ol v_r=-h_{rs}b^s{}_j^\la \cS^j_\la, \qquad v^r=-\frac12
c^r_{pq}c^pc^q.
\ee

In particular, let us turn to Yang -- Mills gauge theory of
principal connections in Section 13. Its constrained Lagrangian
$L_\Pi$ (\ref{yyy10}) is invariant under the gauge transformations
(\ref{ps68}). In view of the transformation law (\ref{ps60}), one
can chose the gauge condition (\ref{ps70}):
\be
g^{\la\m}S^r_{\la\m}(x)-\al^r(x)= \frac12g^{\la\m}(\dr_\la
a^r_\m(x) +\dr_\m a^r_\la(x))-\al^r(x)=0,
\ee
which is the familiar generalized Lorentz gauge. The corresponding
second-order differential operator (\ref{m56}) reads
\be
M^r_s\xi^s=g^{\la\m}(\frac12 c^r_{pq} (\dr_\la a^r_\m +\dr_\m
a^r_\la)\xi^q +c^r_{pq}a^p_\m\dr_\la\xi^q +\dr_\la\dr_\m\xi^r).
\ee
Passing to the Euclidean space and repeating the above
quantization procedure, we come to the generating functional
\be
&& Z=\cN^{-1}\int\exp\{\int (p^{\la\m}_r\cF^r_{\la\m}- a^{mn}_Gg_{\mu\nu}
g_{\la\beta} p^{\mu\la}_m p^{\nu\bt}_n\sqrt{|g|}-\\
&& \qquad \frac18 a^G_{rs}g^{\al\nu}g^{\la\m}(\dr_\al a^r_\nu +\dr_\nu a^r_\al)
(\dr_\la a^s_\m +\dr_\m a^s_\la) -g^{\la\m}\ol c_r(\frac12
c^r_{pq} (\dr_\la a^r_\m
+\dr_\m a^r_\la)c^q+c^r_{pq}a^p_\m c^q_\la +c_{\la\m}^r)  \\
&& \qquad + iJ_r^\mu a^r_\mu +iJ^r_{\m\la} p_r^{\m\la})\om\} \op\prod_x
[d\ol c][d c][dp(x)][da(x)].
\ee
Its integration with respect to momenta restarts the familiar
generating functional of gauge theory.

\section{Algebraic quantization. Quantum PS bracket}

Canonical quantization of time-dependent non-relativistic
mechanics on a fibre bundle $Q\to \mathbb R$ in Section 7 is
adequately formulated is geometric quantization of the vertical
Poisson bracket $\{,\}_V$ (\ref{m72})
\cite{quant2002,book05,book10}. This fact motivates us to
investigate quantization of the PS bracket $\{,\}_{PS}$
(\ref{xx3}).

Let us note that one can quantize only linear spaces of variables,
and therefore the condition of $Y\to X$ to be a vector bundle is
not a loss of generality.

In particular, in order to quantize the PS bracket $\{,\}_{PS}$,
one can be based on the fact that this bracket defines the Lie
bracket (\ref{ps45}) of Noether Hamiltonian currents which brings
a vector space $\cJ(\Pi)$ of these currents into a Lie algebra.
Then a representation of this algebra by operators acting in some
space can be treated as a variant of quantization of the PS
bracket $\{,\}_{PS}$. However, such kind quantization fails to be
a quantization of fields, but that of currents in the spirit of
the well known current algebra approach \cite{current}.

In a different way, we can restrict our consideration to a
subspace of linear functions in $y^i$ and $p^\la_i$ represented as
$(n-1)$-forms $F^\la \om_\la$ due to the isomorphism (\ref{000})
and can modify the PS bracket $\{,\}_{PS}$ (\ref{xx3}) as
\be
\{F,G\}'=\op\int_W \{F,G\}_{PS}
\ee
where $W$ is some compact $(n-1)$-dimensional submanifold of $X$.
This bracket leads us to a nuclear algebra of canonical
commutation relations whose representations can be investigated in
a standard way \cite{book05,sard2002}.

\end{document}